\documentclass{article}

\usepackage{arxiv}

\RequirePackage{amsthm,amsmath,amsfonts,amssymb}
\RequirePackage[numbers]{natbib}
\usepackage{amsmath}
\usepackage{color, graphicx}
\usepackage[colorinlistoftodos]{todonotes}
\usepackage{xcolor}
\usepackage{verbatim}
\usepackage{natbib}
\usepackage{color,soul}
\usepackage{amsthm}
\usepackage{amsmath}
\usepackage{url} 
\usepackage{tikz}
\usepackage{etoolbox}
\usepackage{float}
\usepackage{amsmath}
\usepackage{float}
\usepackage{siunitx}
\usepackage{amssymb}
\usepackage{scalerel}
\usepackage{comment}
\usepackage{graphicx}
\usepackage[parfill]{parskip}
\usepackage{marvosym}
\usepackage{wrapfig}
\usepackage{resizegather}
\usepackage{adjustbox}
\usepackage[toc,page]{appendix}
\usepackage{bbm}
\usepackage{graphicx}
\usepackage{amsbsy}
\usepackage{lipsum}
\usepackage{multicol}
\usepackage{amsthm}
\usepackage{pifont}
\usepackage{bbm}
\usepackage{bm}
\usepackage{dsfont}
\usepackage{bigints}
\usepackage{mathtools}
\usepackage{graphicx}
\usepackage{svg}
\usepackage{calc}
\usepackage{accents}
\usepackage{mwe}
\usepackage{amsmath}
\usepackage{amsthm}
\usepackage{enumitem}
\usepackage{graphicx,wrapfig,lipsum}
\usepackage{mathtools}
\usepackage{mathrsfs}
\usetikzlibrary{mindmap,backgrounds}
\usepackage{blindtext}
\usepackage[hidelinks]{hyperref} % First, ensure no hidden links
\hypersetup{
    colorlinks=false,   % Disable coloring of link text
    pdfborder={0 0 1},  % Creates a visible box around links
    linkbordercolor={1 0 0}, % Red box for internal links
    citebordercolor={0 1 0}, % Green box for citations
    urlbordercolor={0.25 1 1}  % Blue box for URLs
}

\usepackage{graphicx}
\usepackage{subcaption} % For subfigures
\usepackage[numbers]{natbib}
\allowdisplaybreaks

\RequirePackage{graphicx}%% uncomment this for including figures

\usepackage[utf8]{inputenc}
\usepackage{amsmath,amsfonts,amssymb,amsthm}

\usepackage{enumitem}
\usepackage{letltxmacro}
\usepackage{nameref,hyperref}
\usepackage[capitalize]{cleveref}
\usepackage{indentfirst}

\numberwithin{equation}{section}

\newcommand{\indep}{\perp \!\!\! \perp}

\newcommand{\oset}[3][0ex]{%
  \mathrel{\mathop{#3}\limits^{
    \vbox to#1{\kern-2\ex@
    \hbox{$\scriptstyle#2$}\vss}}}}

\usepackage{enumitem}
\usepackage{letltxmacro}
\usepackage{nameref,hyperref}
\usepackage[capitalize]{cleveref}

\newlist{thmlist}{enumerate}{1}
\setlist[thmlist]{label=\alph{thmlisti}., ref=\thetheorem.\alph{thmlisti},noitemsep}

\newcommand{\pb}[0]{\mathrm{P}}
\newcommand{\ex}[0]{\mathrm{E}}
\newcommand{\ind}[0]{\mathbf{1}}

\DeclareMathOperator*{\argmin}{arg\,min}

\newcommand\numberthis{\addtocounter{equation}{1}\tag{\theequation}}
\theoremstyle{plain}
\newtheorem{theorem}{Theorem}%[section]
\newtheorem{corollary}{Corollary}
\newtheorem{lemma}{Lemma}

\theoremstyle{remark}

\newtheorem{condition}{Condition}

\newtheorem*{definition*}{Definition}

\title{Semiparametric Uncertainty Quantification via Isotonized Posterior for  Deconvolutions}
\newcommand{\university}{\textit{Delft University of Technology, Mekelweg 4, Delft 2628CD, The Netherlands.}}
\newcommand{\authname}[1]{\textbf{#1}}
\author{
    \authname{Francesco Gili} \thanks{ \texttt{F.Gili@tudelft.nl}} \quad \quad \quad \authname{Geurt Jongbloed} \thanks{\texttt{G.Jongbloed@tudelft.nl}} \\
    \\
    \university % The university is placed below the authors
}
\hypersetup{
pdftitle={Semiparametric Uncertainty Quantification via Isotonized Posterior for Deconvolutions - F. Gili, G. Jongbloed, A. van der Vaart},
pdfsubject={},
pdfauthor={Francesco ~Gili, Geurt ~Jongbloed},
pdfkeywords={Deconvolution problems, Bayesian Isotonic estimation, Projection Posterior, Dirichlet Process},
}

\begin{document}
\maketitle

\begin{abstract}
	We address the problem of uncertainty quantification for the deconvolution model \(Z = X + Y\), where \(X\) and \(Y\) are nonnegative random variables and the goal is to estimate the signal's distribution of \(X \sim F_0\) supported on~\([0,\infty)\), from observations where the noise distribution is known. Existing frequentist methods often produce confidence intervals for $F_0(x)$ that depend on unknown nuisance parameters, such as the density of \(X\) and its derivative, which are difficult to estimate in practice. This paper introduces a novel and computationally efficient nonparametric Bayesian approach, based on projecting the posterior, to overcome this limitation. Our method leverages the solution \(p\) to a specific Volterra integral equation as in \cite{74}, which relates the cumulative distribution function (CDF) of the signal, \(F_0\), to the distribution of the observables. We place a Dirichlet Process prior directly on the distribution of the observed data $Z$, yielding a simple, conjugate posterior. To ensure the resulting estimates for \(F_0\) are valid CDFs, we isotonize posterior draws taking the Greatest Convex Majorant of the primitive of the posterior draws and defining what we term the Isotonic Inverse Posterior. We show that this framework yields posterior credible sets for \(F_0\) that are not only computationally fast to generate but also possess asymptotically correct frequentist coverage after a straightforward recalibration technique for the so-called Bayes Chernoff distribution introduced in \cite{54}. Our approach thus does not require the estimation of nuisance parameters to deliver uncertainty quantification for the parameter of interest $F_0(x)$. The practical effectiveness and robustness of the method are demonstrated through a simulation study with various noise distributions for $Y$.
\end{abstract}

% keywords can be removed
%\newline
%\newline
%\textbf{MSC2010 subject classification:} 62G05, 62G20, 62C20, 62E20}
\section{Introduction}

In many statistical applications, one encounters situations where the observed data $Z_i$ arises from the sum of two independent components: a variable of interest $X_i$, which is usually the quantity we are interested in, and another variable $Y_i$ whose distribution is assumed to be known or to be inferred using expert knowledge. 

This setup applies to many different areas of research. For example, $X_i$ might represent the true underlying quantity we aim to measure in an experiment or measurement, while $Y_i$ captures the random measurement error introduced during the process --- and which is assumed to be additive in nature. In epidemiology, $X_i$ could correspond to the actual time of infection, and $Y_i$ to the incubation period. This framework plays a key role in approaches like back-calculation, which has been used in HIV/AIDS research to estimate historical infection rates based on observed diagnoses \cite{61}. The statistical properties of $Y_i$ vary significantly depending on the context: measurement errors are often modeled with symmetric distributions over the entire real line, such as the normal distribution, while durations --- like incubation periods or the time until a component fails --- are typically modeled with skewed, non-negative distributions, such as the Gamma or Weibull distributions.

This method of separating the distribution of $X_i$ from the one of $Y_i$, which is called deconvolution, is also key in other fields. For example, in image processing, a blurry image $Z$ can be thought of as the true, sharp image $X$ distorted by a blur effect $Y$ caused by the camera or optics. Methods like the Richardson-Lucy algorithm are used to try and recover the original sharp image $X$ \cite{68, 67}. Similarly, in astronomical spectroscopy, the spectrum of light observed from a star or galaxy $Z$ is the object's actual spectrum $X$ that has been spread out or blurred by the telescope's instruments represented by additive noise $Y$. Deconvolution helps scientists recover the true details about the celestial object \cite{107}. Another application is in seismology, where the seismic waves recorded from underground $Z$ are a combination of the Earth's rock layer structure $X$ and the seismic signal $Y$ that was sent out or generated. Deconvolution is thus a fundamental step to get clearer images of the Earth's subsurface \cite{109}.

More formally, throughout this paper we confine attention to the case where both $X$ and $Y$ are nonneg\-ative random variables, i.e.\ we work on the half-line $[0,\infty)$. Let $X_1, X_2, \dots, X_n$ denote a sample from an unknown distribution with distribution function $F_0$ supported on $[0, \infty)$ and, independent of that sample, $Y_1, Y_2, \dots, Y_n$ a sample from a known distribution with density $k$ on $[0, \infty)$. Consider the problem of estimating $F_0$ based on the sample $Z_1, Z_2, \dots, Z_n$, where $Z_i = X_i + Y_i$. The density $g_0$ of $Z_i$ is the convolution of $k$ and $F_0$ in the following sense:
\[
g_0(z) = \int_{0}^{\infty} k(z - x) \, dF_0(x) =: k \ast dF_0(z).
\]
For this reason, this estimation problem is known as a deconvolution problem.

Suppose that, given a kernel $k$, we have a function $p$ living on $[0, \infty)$, solving the integral equation
\begin{equation} \label{eq:1}
(p \ast k)(x) := \int_0^x p(x-y)k(y) \, dy = (\ind{} \ast \ind{})(x) = x\ind{}(x),
\end{equation}
where the function $\ind{}$ is defined by
\[
\ind{}(x) = \ind_{[0, \infty)}(x).
\]
Then we can write, for each $x > 0$ and a $Z$ having density function $g_0 = k \ast dF_0$,
\begin{equation} \label{eq:2}
\ex[p(x-Z)] = (p \ast g_0)(x) = (p \ast k \ast dF_0)(x) = (\ind{} \ast \ind{} \ast dF_0)(x) = \int_0^x F_0(s) \, ds =: H_0 (x).
\end{equation}
A crucial question concerning integral equation \eqref{eq:1} is whether a function $p$ exists that satisfies it. In Appendix \ref{sec: review integral equation}, we give a review of the literature on this topic.

We start by placing a Dirichlet Process prior directly on the distribution function of the observables: \( G \sim \operatorname{DP}(\alpha) \). Given an observed sample \( Z_1, \dots, Z_n \) from cdf \( G_0 \) (associated with density $g_0$), the posterior distribution is (c.f.\ \cite{59}, Theorem 4.6 in \cite{4}):
\begin{align}\label{eq: dirichlet posterior}
	G \mid Z_1,\ldots,Z_n \sim \text{DP}(\alpha + n \mathbb{G}_n), \quad \quad  \mathbb{G}_n := \frac{1}{n} \sum_{i=1}^n \delta_{Z_i}.
\end{align}
This choice offers computational simplicity due to the conjugacy of the posterior. Leveraging the inverse relations inherent to the problem, we define, by analogy with \eqref{eq:2}, the functional $H_{\scriptscriptstyle{G}}$ (with $G$ in place of $G_0$) as
\begin{align}\label{eq: naive bayes}
	H_{\scriptscriptstyle{G}}(x)=\int_{[0, x)} p(x-z) \, d G(z).
\end{align}
Similarly, the empirical counterpart of the left-hand side of \eqref{eq:2} is given by a sample mean:
$$
H_n(x)=\int_{[0, x)} p(x-z) \, d \mathbb{G}_n(z)=\frac{1}{n} \sum_{i=1}^n p\left(x-Z_i\right)
$$
This function $H_n$ is an estimator for $H_0$. 

Taking the derivative of some smoothed version of $H_n$ or $H_{\scriptscriptstyle{G}}$ ($H_n$ and $H_{\scriptscriptstyle{G}}$ will in general not be differentiable) would therefore yield methodologies for the estimation of $F_0$. We call such a methodology an inverse methodology, since it is based on the inverse relation
$$
F_0(x)=\frac{d}{d x} p * g_0(x) \quad \text { a.e. }
$$
which follows from \eqref{eq:2}. However, by using general smoothing techniques, for example kernel estimation, we would not use the information that $H_0$ is convex (which follows from \eqref{eq:2} and the monotonicity of $F_0$). Consequently, inverse estimators based on standard smoothing techniques will in general yield non-monotone estimators. 

To remedy this, we isotonize the posterior draws using the following procedure. For $T \in (0, \infty]$ and a continuous function $h : [0,\infty) \to [0,\infty)$, the Greatest Convex Minorant (GCM) of $h$ on $[0,T]$ and its right-hand side derivative, are given for $x \in [0,T]$ by:
\begin{align*}
    &h^{\star}(x) := \sup {\{ f(x) \, : \, f \leq h \, \, \text{on} \, \, [0,T], \, f \, \text{convex} \, \, \text{on} \, \, [0,T] \}},\\
    & (h^{\star})^{\prime}_+ (x):= \lim_{t \to 0^+} \frac{h^{\star}(x + t) - h^{\star}(x)}{t}.
\end{align*}
For any $x \in [0, T]$, we define the isotonizations of the draws of the posterior as the right derivative of $H_{\scriptscriptstyle{G}}^{\star}$ evaluated at $x$,
\begin{align}\label{eq: def IIP}
	\hat{F}_{\scriptscriptstyle{G}}(x) = (H_{\scriptscriptstyle{G}}^{\star})^{\prime}_+ (x).
\end{align}
The posterior draws $\hat{F}_{\scriptscriptstyle{G}}$ are now by construction monotone (isotonic with respect to natural ordering on $\mathbb{R}$), and therefore we call the posterior \textit{Isotonic Inverse Posterior}. Using the same procedure we define the \textit{Isotonic Inverse Estimator} of \cite{74}, for any $x \in [0, T]$, as:
\begin{align}\label{eq: def IIE}
	\hat{F}_n(x) = (H_n^{\star})^{\prime}_+ (x).
\end{align}
One possible choice for $T$ in the definition of $\hat{F}_{\scriptscriptstyle{G}}$ or $\hat{F}_n$ is $T=\infty$. As we will see in Section \ref{sec: main results}, we need finiteness of $T$ in order to prove our asymptotic distribution result for a large class of densities $k$. If $k$ is assumed monotone on $[0,\infty)$, our result does allow the choice $T=\infty$. For practical purposes, there is no difference between finite (but large) and infinite $T$. 

In the domain of nonparametric inference, frequentist methods often yield estimators with well-characterized asymptotic distributions (e.g., the Chernoff distribution for isotonic regression, c.f. \cite{2}). These distributions are crucial for constructing asymptotic confidence intervals and understanding the estimator's behavior. However, a significant practical limitation arises because these limiting distributions generally depend on unknown nuisance parameters. For instance, in our deconvolution problem of estimating $F_0(x)$ from samples $Z_1, \dots, Z_n \sim g_0 = k * dF_0$, the limiting distribution depends on quantities such as $g_0(x)$ and, more critically, $f_0(x)$ (in a similar way as in the isotonic regression setting where the limit distribution often involves the derivative of the true underlying function). The necessity of estimating these nuisance parameters complicates the direct application of these asymptotic results for uncertainty quantification. This practical difficulty highlights the desirability of methods that can provide uncertainty quantification without relying on the estimation of such nuisance parameters. Bayesian nonparametric procedures offer a promising alternative in this regard. As explored in approaches in the ``projection-posterior" literature (c.f. \cite{54,55,56}), a key advantage is that the resulting credible intervals can achieve asymptotic frequentist coverage that is free of these nuisance parameters. The limiting coverage in such Bayesian frameworks can depend only on the chosen credibility level, not on characteristics of the true unknown function like its derivative $f_0(x)$. This circumvents the harder problem of nuisance parameter estimation, leading to more direct and more robust inferential procedures for quantifying uncertainty in nonparametric models. Furthermore, another possible Bayesian approach would consist in placing a DP process prior directly on $F$. However this procedure has some clear computational disadvantages due to the needed Gibbs sampling procedures to sample from the posterior.

Thus the approach in this paper has two advantages:
\begin{itemize}
	\item[(i)] \textbf{Computational speed}: The Dirichlet posterior \eqref{eq: dirichlet posterior} is explicit and its projection can be effectively computed by the Pool-Adjacent-Violators Algorithm (PAVA, c.f.\ \cite{61}).  
	\item[(ii)] \textbf{Uncertainty quantification}: The Bayesian framework inherently provides uncertainty quantification without the need to estimate nuisance parameters. We show below that this is asymptotically correct in the frequentist sense, after an appropriate calibration.
\end{itemize}
Advantage (ii) is the core contribution of this paper, as currently there are no nuisance-free methods in the literature for semiparametric uncertainty quantification of $F_0(x)$ for deconvolution problems as treated in this paper. Projecting the unconstrained posterior distribution departs from a classical Bayesian framework which might assign a Dirichlet Process (DP) prior to \( F \) and employ a hierarchical model to sample from the density of the observations. Uncertainty quantification results for this approach are currently unavailable and, at the same time, this approach is not advantageous from a computational prespective, being slower as described above.

\subsubsection*{Motivation and connections with the literature} 

There is a vast literature on deconvolution problems, especially in the frequentist nonparametric literature, where the nonparametric maximum likelihood estimator (NPMLE) is a central focus. For the special case where the kernel $k$ is a decreasing density on $[0, \infty)$ and $F_0(0) = 0$, \cite{71} study the NPMLE for $F_0$, demonstrating its consistency and conjecturing its asymptotic distribution. They also introduce the broader problem of one-sided errors in convolution models, deriving NPMLEs for the c.d.f. in this context. While these NPMLEs often lack explicit expressions and require iterative computation, some particular cases yield more direct solutions. For instance, uniform deconvolution is addressed by \cite{98, 99}, and more recently by \cite{75} who propose kernel-based estimators and inversion formulas. Exponential deconvolution is tackled by \cite{104, 73}. For situations lacking such explicit forms, \cite{74} circumvent this by proposing an isotonic inverse estimator. Many of these maximum likelihood estimators, and the works surrounding them, share cube root asymptotics similar to the Grenander estimator for decreasing densities and NPMLEs in current status censoring problems (c.f. \cite{2}), with a general focus on deriving the asymptotic distributions of these monotonic estimators. An alternative approach to deconvolution utilizes the convolution structure for inversion via Fourier transform techniques. Kernel estimators based on this methodology are introduced by several authors (c.f. \cite{84} and \cite{101}), using direct inversion formulas for gamma and Laplace deconvolution are found in \cite{102}. While these kernel-based methods can achieve faster convergence rates if the unknown $F_0$ is smooth, a notable disadvantage is that their resulting estimators for $F_0$ are not inherently monotone, unlike the NPMLE and isotonic inverse estimators.

Within the broader field of deconvolution, a lot of work is done to address the estimation of a signal distribution $f$ observed with additive, known-distribution noise. This noise generally is conceptualized as centered, differentiating it from one-sided error scenarios. Foundational studies on convergence rates and optimality for kernel-based estimators include those by \cite{81}, \cite{82}, \cite{83}, \cite{84}, and \cite{85}. Establishing sharp asymptotic optimality further is done by \cite{86} and \cite{87,88}. Adaptive methodologies, particularly for bandwidth selection in models with known error distributions, also are extensively developed. Some examples include wavelet-based methods \cite{89}, direct bandwidth selection techniques \cite{90}, and penalized projection strategies \cite{91}, with comprehensive treatments available in works like \cite{92}. The estimation of the cumulative distribution function (c.d.f.) in convolution models also receives attention. Contributions from \cite{93}, \cite{84}, \cite{94}, \cite{95}, \cite{96}, and \cite{97} primarily focus on pointwise estimation procedures, a consequence of the c.d.f. not being square-integrable on $\mathbb{R}$. The latter two studies address c.d.f. estimation with unknown error distributions, contingent on specific tail decay properties of the error's characteristic function, and achieve optimal rates for Sobolev-class target functions. It is worth noting that the deconvolution of non-negative variables also emerges in applied contexts such as actuarial and insurance modeling. Finally, some financial applications see contributions from \cite{79} and \cite{80} addressing one-sided errors, focusing on optimal adaptive estimation in non-parametric regression where errors (often exponential) are not assumed to be centered.

A similar methodology to the one proposed in this paper — using an unconstrained conjugate prior for the distribution function of the observables and projecting the posterior samples into the space of monotone functions — very recently is adopted in isotonic regression settings (c.f.\ \cite{54,55,56}), isotonic density estimation (c.f.\ \cite{60}), ODEs settings (c.f.\ \cite{57,58}) and for Wicksell's problem (\cite{50},\cite{65}). \cite{57,58} develop Bayesian two-step methods for parameter inference in ODE models by embedding them in a nonparametric regression framework with B-spline priors. The parameter posterior is obtained by minimizing a discrepancy between the estimated regression function and the ODE solution, either via derivative matching or Runge–Kutta approximations. \cite{58} further establishes a Bernstein–von Mises theorem for the Runge–Kutta-based method, showing that the resulting Bayes estimator is asymptotically efficient. However, their framework relies on a very different type of prior modelling and the projection method does not entail isotonization. In \cite{54}, the authors study a nonparametric Bayesian regression model for a response variable \( Y \) with respect to a predictor variable \( X \in [0,1] \), given by \( Y = f(X) + \varepsilon \), where \( f \) is a monotone increasing function and \( \varepsilon \) is a mean-zero random error. They propose a "projection-posterior" estimator \( f^* \), constructed using a finite random series of step functions with normal basis coefficients as a prior for \( f \). Their main finding is that when \( f^* \) is centered at the maximum likelihood estimator (MLE) and rescaled appropriately, the Bernstein-von Mises (BvM) theorem does not hold. However, when \( f^* \) is instead centered at the true function \( f_0 \), a BvM result is obtained. This phenomenon resonates with the known inconsistency of the bootstrap for the Grenander estimator in isotonic regression (c.f.\ \cite{63}) and represents a key similarity with the present work.

One of the key advantages our approach shares with the works in \cite{57,58,54,55,56} is the computational efficiency, driven by posterior conjugacy, as well as the relative simplicity of the asymptotic analysis.

\section{Main results}\label{sec: main results}

In this section, we present the main result of this paper, which establishes the asymptotic distribution of the isotonized posterior of $\hat{F}_{\scriptscriptstyle{G}}$, and we compare it with the asymptotic distribution attained by the isotonized inverse estimator (IIE) $\hat{F}_n$ (c.f.\ \cite{74}) for a fixed $x \in [0, T)$, where $T$ is a positive constant. We show in particular an interesting analogy. The results of \cite{74} show that the IIE $\hat{F}_n$
attains at cube $n$ root rate an asymptotic distribution, which, up to constants, is the Chernoff distribution. Symmetrically, we show that under the same conditions, the isotonized posterior of $\hat{F}_{\scriptscriptstyle{G}}$ attains at cube root rate an asymptotic distribution which, up to constants, is the Bayesian Chernoff distribution \eqref{eq:ZB_def}, as introduced in \cite{54}. This curious phenomenon both confirms and extends the vision of \cite{54}, who introduced the Bayesian Chernoff distribution \eqref{eq:ZB_def}, proved its key properties (symmetry, monotone CDF, and the resulting recalibration procedure), and predicted that this distribution would play a central role in uncertainty quantification in nonparametric Bayesian inference far beyond the regression setting they originally considered. Indeed, also in the current setting, with a very different inverse problem and different prior, the Bayesian Chernoff distribution arises as the limiting distribution of the isotonized posterior of $\hat{F}_{\scriptscriptstyle{G}}$, confirming that the concept and the framework of \cite{54} have a broad reach.

Integral equation \eqref{eq:1} is a Volterra equation of the first kind, of convolution type. The function $p$ is sometimes called the \textit{resolvent of the first kind} of $k$ (see \cite{103}, page 158). In all what follows, we restrict our attention to the case where $F_0$ has support contained in $[0, \infty)$. Furthermore, to prove consistency of the posterior $\hat{F}_{\scriptscriptstyle{G}}$, we have to impose a condition on $p$. The asymptotic distribution is derived under more stringent conditions (the same as in \cite{74}) that the kernel $k$ satisfies certain regularity conditions, which we will specify below.

\begin{condition} \label{cond:1}
On bounded intervals, the function $p$ has only finitely many discontinuities. All these discontinuities are finite in size.
\end{condition}

For the cube root asymptotics of $\hat{F}_{\scriptscriptstyle{G}}(x)$ for $x < T$, we need a slightly stronger condition on $p$

\begin{condition} \label{cond:2}
The function $p$ is Hölder continuous of order $\gamma > 1/2$ on $(0, \infty)$ and $0 < p(0) < \infty$ and right continuous at zero.
\end{condition}

A crucial question concerning integral equation \eqref{eq:1} is whether a function $p$ exists that satisfies it. In Appendix \ref{sec: review integral equation}, we give a review of the literature on this topic, showing that such a function $p$ exists for many kernels $k$ of interest. We also report a lemma, Lemma \ref{lemma:1}, that provides a sufficient condition for the existence of such a function $p$ that satisfies the main condition required by the results in \cite{74} and thus by our results, Condition \ref{cond:2}.

For convenience of the reader, we report here without proof the main result of \cite{74} for the IIE $\hat{F}_n$ (c.f.\ Theorem 2 in \cite{74}).

\begin{theorem}[Theorem 2 in \cite{74}]\label{thm: frequentist UQ}
Let $p$ satisfy Condition \ref{cond:2}, $0 < T < \infty$ and let $x \in [0, T)$ be fixed and $F_0$ be such that $F_0$ has a continuous strictly positive derivative $f_0$ in a neighborhood of $x$. Then $n^{1/3}(\hat{F}_n(x) - F_0(x))$ converges in distribution as $n \to \infty$; specifically, $\forall \, z \in \mathbb{R}$:
\begin{align}
&\pb \left( n^{1/3} (\hat{F}_n(x) - F_0(x)) \le z \right) \xrightarrow{d} \pb \left( \frac{2^{2/3} f_0(x)^{1/3} g_0(x)^{1/3}}{k(0)^{2/3}} \argmin_{t \in \mathbb{R}} \{\mathbb{W}_1(t) + t^2 \} \le z \right)
\end{align}
where $\mathbb{W}_1$ is two-sided Brownian motion originating from $0$.
\end{theorem}

Finally, we proceed to the main result of this paper, which establishes the asymptotic distribution of the isotonized posterior of $\hat{F}_{\scriptscriptstyle{G}}$. By stating it next to Theorem \ref{thm: frequentist UQ}, we highlight the analogy between the asymptotic distribution of the IIP and that of the IIE.

\begin{theorem}\label{thm: bayesian UQ}
Suppose $G \sim \mathrm{DP}(\alpha)$, where the base measure $\alpha$ has bounded density, and $\hat{F}_{\scriptscriptstyle{G}}$ as in \eqref{eq: def IIP}. Let $p$ satisfy Condition \ref{cond:2}, $0 < T < \infty$ and let $x \in [0, T)$ be fixed and $F_0$ be such that $F_0$ has a continuous strictly positive derivative $f_0$ in a neighborhood of $x$. Then $n^{1/3}(\hat{F}_{\scriptscriptstyle{G}}(x) - F_0(x))$ converges in distribution conditionally on $Z_1, \dots, Z_n$ as $n \to \infty$; specifically, $\forall \, z \in \mathbb{R}$:
\begin{align}
&\Pi \left( n^{1/3} (\hat{F}_{\scriptscriptstyle{G}}(x) - F_0(x)) \le z \mid Z_1, \dots, Z_n \right) \\
&\xrightarrow[n \to \infty]{d}  \pb \left( \frac{2^{2/3} f_0(x)^{1/3} g_0(x)^{1/3}}{k(0)^{2/3}} \argmin_{t \in \mathbb{R}} \{\mathbb{W}_1(t) + \mathbb{W}_2(t) + t^2 \} \le z \mid \mathbb{W}_1 \right),
\end{align}
where $\mathbb{W}_1$ and $\mathbb{W}_2$ are two independent Brownian motions on $\mathbb{R}$ originating at $0$.
\end{theorem}

\begin{proof}
Consider for $a \in (0,1)$ and $\tau \in [0,T)$ the event $T_{\scriptscriptstyle{G}}(a) > \tau$ where:
\begin{align}\label{def: T_G}
T_{\scriptscriptstyle{G}}(a) := \inf \{ t \in [0,T] : H_{\scriptscriptstyle{G}}(t) - a t \, \, \text{is minimal} \}.
\end{align}
This event takes place if and only if the maximal affine function with slope $a$ that is bounded above by $H_{\scriptscriptstyle{G}}$ coincides with $H_{\scriptscriptstyle{G}}$ at a point $t_0 > \tau$, while for all $t \le \tau$, it stays strictly beneath $H_{\scriptscriptstyle{G}}$. This is equivalent to the inequality $\hat{F}_{\scriptscriptstyle{G}}(\tau) < a$. So $\forall \, \tau \in [0,T)$ and $a \in (0,1):$
\[
T_{\scriptscriptstyle{G}}(a) \le \tau \iff \hat{F}_{\scriptscriptstyle{G}}(\tau) \ge a.
\]
Fix $x \in (0,T)$. Then, by Lemma \ref{lemma: rewriting1}, for $a \in \mathbb{R}$ and $n$ sufficiently large
\begin{align*}
n^{1/3} (\hat{F}_{\scriptscriptstyle{G}}(x) - F_0(x)) < a \iff \inf \{ t \in [-x n^{1/3}, (T-x) n^{1/3}] : Z_{\scriptscriptstyle{G}}(t) - a t \, \, \text{is minimal} \} > 0,
\end{align*}
where
\begin{align}\label{def: Z_G}
Z_{\scriptscriptstyle{G}}(t) := n^{2/3} \left( H_{\scriptscriptstyle{G}}(x + t n^{-1/3}) - H_{\scriptscriptstyle{G}}(x) - F_0(x) t n^{-1/3} \right).
\end{align}
This process can be decomposed as
\[
Z_{\scriptscriptstyle{G}}(t) = W_{\scriptscriptstyle{G}}(t) + W_n(t) + n^{2/3} (H_0(x + t n^{-1/3}) - H_0(x) - F_0(x) t n^{-1/3})
\]
where:
\begin{align*}
W_{\scriptscriptstyle{G}}(t) &:= n^{2/3} (H_{\scriptscriptstyle{G}}(x + n^{-1/3}t) - H_{\scriptscriptstyle{G}}(x) - H_n(x + n^{-1/3}t) + H_n(x)), \\
W_n(t) &:= n^{2/3} (H_{n}(x + n^{-1/3}t) - H_{n}(x) - H_0(x + n^{-1/3}t) + H_0(x)).
\end{align*}
We can rewrite $W_n(t)$ by identifying two processes: one of them converging weakly and a remainder term that converges to zero in probability. By defining the $\gamma$-Hölder continuous, $\gamma > 1/2$, function $\tilde{p} := p - p(0) \ind_{[0,\infty)}$, and 
\begin{align*}
\mathbb{W}_n(t) &:= n^{2/3} p(0) \int_0^\infty (\ind_{[0, x + n^{-1/3}t)}(z) - \ind_{[0, x)}(z)) \, d(\mathbb{G}_n - G_0)(z), \\
R_n(t) &:= n^{2/3} \int_0^\infty (\tilde{p}(x + n^{-1/3}t - z) - \tilde{p}(x - z)) \, d(\mathbb{G}_n - G_0)(z),
\end{align*}
we can write:
\begin{align}\label{eq: W_n decomposition}
	W_n(t) = \mathbb{W}_n(t) + R_n(t).
\end{align}
Using the fact $\tilde{p} = p - p(0) \ind_{[0,\infty)}$ is Hölder continuous of degree $\gamma > 1/2$ we have that as $n \to \infty$:
\begin{align}\label{eq: residual term convergence}
	\sup_{|t| \in K} |R_n(t)| \xrightarrow{G_0} 0 
\end{align}
for all compact sets $K \subset \mathbb{R}$ (see appendix in \cite{74} for a proof). By example 3.2.14. in \cite{5}, for every compact $K \subset \mathbb{R}$:
\begin{align}\label{eq: mathbbW_n convergence}
	\mathbb{W}_n \rightsquigarrow \frac{\sqrt{g_0(x)}}{k(0)} \mathbb{W}_1 \quad \text{in } \, \, \ell^\infty(K)
\end{align}
where $\mathbb{W}_1$ is a two sided standard Brownian Motion originating at $0$. Now, using Taylor's theorem, for $n \to \infty$
\begin{align}\label{eq: Taylor expansion}
	n^{2/3} (H_0(x + n^{-1/3}t) - H_0(x) - F_0(x) n^{-1/3}t) = \frac{1}{2} f_0(x) t^2 + o(1),
\end{align}
uniformly for $t$ in compacta. 

We now turn to studying $W_{\scriptscriptstyle{G}}$. Define $f_t^n(z) = n^{1/3} (p(x+n^{-1/3}t-z)\ind_{[0, x+n^{-1/3}t)}(z) - p(x-z)\ind_{[0, x)}(z) )$
then, for independent $Q \sim \text{DP}(\alpha)$, $\mathbb{B}_n \sim \text{DP}(n \mathbb{G}_n)$ and $V_n \sim \text{Be}(|\alpha|, n)$:
\[
G f_t^n = V_n Q f_t^n + (1-V_n) \mathbb{B}_n f_t^n
\]
and thus we study the process:
\[
W_{\scriptscriptstyle{G}}(t) = n^{1/3}(G-\mathbb{G}_n)f_t^n = n^{1/3} V_n (Q-\mathbb{G}_n)f_t^n + n^{1/3}(1-V_n)(\mathbb{B}_n-\mathbb{G}_n)f_t^n.
\]
By Lemma \ref{lemma: simplifications} (eq.\ \eqref{eq: V_n term 1}-\eqref{eq: V_n term 2}), the terms premultiplied by $V_n$ converge to zero in probability under $G_0$, so we can focus on the study of:
\[
W_n^*(t) := n^{1/3} (\mathbb{B}_n-\mathbb{G}_n) f_t^n.
\]
By Proposition G.2 in \cite{4} for \( \varepsilon_i \overset{\text{i.i.d.}}{\sim} \text{Exp}(1) \) and \( W_{ni} := \varepsilon_i / \bar{\varepsilon}_n \), we have \( (W_{n1}, \ldots, W_{nn}) \sim \text{Dir}(n; 1/n, \dots, 1/n) \). Since \( \mathbb{B}_n \sim \text{DP}(n \mathbb{G}_n) \), denoting by $\delta_{Z_i}$ Dirac measure at $Z_i$, we have the following representation in distribution of $\mathbb{B}_n$ (c.f.\ Example 3.7.9 in \cite{5}):
	\[
	\mathbb{B}_n = \frac{1}{n} \sum_{i=1}^n W_{ni} \delta_{Z_i}.
	\]
Because of this representation (c.f.\ section 3.7.2. in \cite{5}), $\mathbb{B}_n$ is also known as \textit{Bayesian bootstrap}. Conclude, by Lemma \ref{lemma: further simplification}, that uniformly for $t$ in compacta:
\begin{align}\label{eq: W_n^* simplified}
	W_n^*(t) = n^{1/3}(\mathbb{B}_n - \mathbb{G}_n)f_t^n &= \frac{1}{n^{2/3}} \sum_{i=1}^n (W_{ni}-1) f_t^n(Z_i) = \mathbb{W}_n^*(t) + o_p(1) .
\end{align}
where for \( \varepsilon_i \overset{\text{i.i.d.}}{\sim} \text{Exp}(1) \), the process $\mathbb{W}_n^*$, for $\tilde{f}_t^n(z) := n^{1/3} p(0) \ind_{[x, x+n^{-1/3}t)}(z)$, is defined as:
\begin{align}\label{eq: mathbbW star process}
    \mathbb{W}_n^*(t) = \frac{1}{\sqrt{n}} \sum_{i=1}^n \frac{(\varepsilon_i-1)}{\sqrt[6]{n}} \tilde{f}_t^n(Z_i).
\end{align}

Using the results in Appendix A, Lemma \ref{lemma: conditional convergence W star}, which shows that $\mathbb{W}_n^*$ converges conditionally in $\ell^\infty(K)$ to $\frac{\sqrt{g_0(x)}}{k(0)} \mathbb{W}_2 $, together with Lemma \ref{lemma: simplifications} and \eqref{eq: W_n^* simplified}, we have in probability under $G_0$, as $n \to \infty$,
\begin{align}\label{eq: W_G convergence}
	W_{\scriptscriptstyle{G}} \mid Z_1, \dots, Z_n \rightsquigarrow \frac{\sqrt{g_0(x)}}{k(0)}\mathbb{W}_2 \quad \text{in } \, \, \ell^\infty(K),
\end{align}
for every compact $K \subset \mathbb{R}$. Hence, by \eqref{eq: W_n decomposition}, \eqref{eq: residual term convergence}, \eqref{eq: mathbbW_n convergence}, \eqref{eq: W_G convergence} and Slutsky's lemma, for every compact $K \subset \mathbb{R}$:
\begin{align*}
& (W_{\scriptscriptstyle{G}} (t) + W_n(t) + n^{2/3} (H_0(x + t n^{-1/3}) - H_0(x) - F_0(x) t n^{-1/3}) : t \in K) \mid  Z_1, \dots, Z_n \\
& \quad \quad \quad \quad \quad \leadsto \left(  \frac{\sqrt{g_0(x)}}{k(0)} \mathbb{W}_1(t) +  \frac{\sqrt{g_0(x)}}{k(0)} \mathbb{W}_2(t) + \frac{1}{2} f_0(x) t^2 : t \in K \right) \mid \mathbb{W}_1. \numberthis \label{eq: general weak convergence}
\end{align*}
Let us elaborate on why the unconditional convergence $W_n \rightsquigarrow \frac{\sqrt{g_0(x)}}{k(0)}\mathbb{W}_1$ in $\ell^\infty(K)$ (equation \eqref{eq: mathbbW_n convergence}) gives rise to a \emph{conditional} limit involving $\mathbb{W}_1$ in \eqref{eq: general weak convergence}. The key observation is that $W_{\scriptscriptstyle{G}}$ is driven exclusively by the Bayesian bootstrap weights $(\varepsilon_i)_{i=1}^n$, which are independent of the data $(Z_1, \dots, Z_n)$. Hence $W_{\scriptscriptstyle{G}}$ and $W_n$ are conditionally independent given the data. As $n \to \infty$: the orbit of $W_n$ (a measurable function of the data alone) converges to $\frac{\sqrt{g_0(x)}}{k(0)}\mathbb{W}_1$, and $W_{\scriptscriptstyle{G}}|Z_1,\dots,Z_n$ converges independently to $\frac{\sqrt{g_0(x)}}{k(0)}\mathbb{W}_2$. In the joint limit, the realization of $\mathbb{W}_1$ (coming from $W_n$) plays the role of a fixed, conditioning object, while $\mathbb{W}_2$ remains a fresh independent Brownian motion. Conditional on the data, an application of Slutsky's lemma in the conditional setting then yields \eqref{eq: general weak convergence} with the right-hand side conditioned on $\mathbb{W}_1$.

Now, in order to apply Lemma A.3 in \cite{65}, we need to prove that there exists a stochastically bounded sequence $\tilde{M}_n$ such that:
\begin{align}\label{eq: conditional stoch boundedness}
	\Pi_n \left( \inf \{ t \in [-x n^{1/3}, (T-x)n^{1/3}] : Z_{\scriptscriptstyle{G}}(t) - a t \, \,  \text{is  minimal} \} \le \tilde{M}_n \mid Z_1, \dots, Z_n \right) \xrightarrow{G_0} 1 
\end{align}
The proof of \eqref{eq: conditional stoch boundedness} is given in Appendix A, Lemma \ref{lemma: conditional stoch boundedness}. An application of Lemma A.3 in \cite{65}, together with \eqref{eq: general weak convergence} and \eqref{eq: conditional stoch boundedness} gives:
\begin{align*}
	&\inf \{ t \in \mathbb{R} : Z_{\scriptscriptstyle{G}}(t) - a t \text{ is minimal} \} \mid Z_1, \dots, Z_n \\
	&\leadsto \argmin_{t \in \mathbb{R}} \left( \frac{\sqrt{g_0(x)}}{k(0)} \mathbb{W}_1(t) + \frac{\sqrt{g_0(x)}}{k(0)} \mathbb{W}_2(t) + \frac{1}{2} f_0(x) t^2 - at \right) \mid \mathbb{W}_1.
\end{align*}
Using that $\forall \, c > 0$, $b \in \mathbb{R}$ that as process the standard Brownian Motion satisfies $W(\cdot/a^2) \stackrel{d}{=} W(\cdot)/a$:
\[ \operatorname{argmin}_{t \in \mathbb{R}} \{ c (\mathbb{W}_1(t) + \mathbb{W}_2(t)) + (t-b)^2 \} \stackrel{d}{=} c^{2/3} \operatorname{argmin}_{t \in \mathbb{R}} \{ \mathbb{W}_1(t) + \mathbb{W}_2(t) + t^2 \} + b. \]
Note that this distributional equality holds \emph{jointly in the pair $(\operatorname{argmin}, \mathbb{W}_1)$}: the time-rescaling $t \mapsto c^{1/3}t$ is applied simultaneously to both $\mathbb{W}_1$ and $\mathbb{W}_2$, so the rescaled $\mathbb{W}_1(c^{1/3}\,\cdot)/c^{1/6}$ has the same law as $\mathbb{W}_1$, and a joint rescaling argument shows that the equality in distribution is in the space of $(\text{argmin}, \mathbb{W}_1)$. Thus conditioning on $\mathbb{W}_1$ on the right-hand side corresponds to conditioning on the same Brownian trajectory (up to the distributional scaling identity), and the conditional distribution of the argmin given $\mathbb{W}_1$ is the same on both sides of the $\stackrel{d}{=}$. Using the change of variable $t \to \sqrt{\frac{2}{f_0(x)}} t$ we get
\begin{align*}
&\argmin_{t \in \mathbb{R}} \left( \frac{\sqrt{g_0(x)}}{k(0)} \mathbb{W}_1(t) + \frac{\sqrt{g_0(x)}}{k(0)} \mathbb{W}_2(t) + \frac{1}{2} f_0(x) t^2 - at \right) \\
&\, \, \, \, \, \stackrel{d}{=} \frac{g_0(x)^{1/3}}{k(0)^{2/3}} \left( \frac{2}{f_0(x)} \right)^{1/6+1/2} \argmin_{t \in \mathbb{R}} \{ \mathbb{W}_1(t) + \mathbb{W}_2(t) + t^2 \} + \frac{a}{f_0(x)},	
\end{align*}
which proves the claim, since:
\begin{align*}
&\pb \left(\argmin_{t \in \mathbb{R}} \left( \frac{\sqrt{g_0(x)}}{k(0)} \mathbb{W}_1(t) + \frac{\sqrt{g_0(x)}}{k(0)} \mathbb{W}_2(t) + \frac{1}{2} f_0(x) t^2 - at \right) > 0 \mid \mathbb{W}_1  \right) \\
&=\pb \left( \frac{g_0(x)^{1/3}}{k(0)^{2/3}} \left( \frac{2}{f_0(x)} \right)^{1/6+1/2} \argmin_{t \in \mathbb{R}} \{ \mathbb{W}_1(t) + \mathbb{W}_2(t) + t^2 \} > - \frac{a}{f_0(x)} \mid \mathbb{W}_1  \right)	\\
&= \pb \left( \frac{2^{2/3} f_0(x)^{1/3} g_0(x)^{1/3}}{k(0)^{2/3}} \argmin_{t \in \mathbb{R}} \{\mathbb{W}_1(t) + \mathbb{W}_2(t) + t^2 \} \le a \mid \mathbb{W}_1 \right).
\end{align*}

\end{proof}

The limiting distribution established in Theorem \ref{thm: bayesian UQ} is fundamental for establishing the asymptotic coverage of credible intervals. For any given $n \ge 1$ and $\tau \in [0, 1]$, the credible interval is represented by $I_{n,\tau} := [Q_{n,1-\tau/2}, Q_{n,\tau/2}]$. Here, $Q_{n,\tau}$ signifies the $(1-\tau)$-quantile of the projection-posterior distribution associated with $\hat{F}_{\scriptscriptstyle{G}}(x)$:
$$ Q_{n,\tau} := \inf \{ z \in \mathbb{R} \, : \, \Pi ( \hat{F}_{\scriptscriptstyle{G}}(x) \le z \mid Z_1, \ldots, Z_n ) \ge 1- \tau \}. $$
Consequently, $F_0(x) \le Q_{n,\tau}$ is true if and only if $\Pi ( \hat{F}_{\scriptscriptstyle{G}}(x) \le F_0(x) \mid Z_1, \ldots, Z_n ) \le 1- \tau$. The limiting coverage of the interval $I_{n,\tau}$ as $n \rightarrow \infty$ can be determined. To begin, observe that
\begin{align}\label{eq: one side coverage}
	\pb (F_0(x) \le Q_{n,\tau}) &= \pb (\Pi(n^{1/3}(\hat{F}_{\scriptscriptstyle{G}}(x) - F_0(x)) \le 0 \mid Z_1,\ldots,Z_n) \le 1-\tau) \\
	&\to \pb \left( \pb (\operatorname{argmin}_{t \in \mathbb{R}}\{ \mathbb{W}_1(t) + \mathbb{W}_2(t) + t^2\} \le 0 \mid \mathbb{W}_1) \le 1-\tau \right).
\end{align}
This can be written as $\pb (Z_B \le 1-\tau)$, where $Z_B$ is specified as
\begin{align} \label{eq:ZB_def}
	Z_B := \pb (\operatorname{argmin}_{t \in \mathbb{R}}\{ \mathbb{W}_1(t) + \mathbb{W}_2(t) + t^2 \} \le 0 \mid \mathbb{W}_1).
\end{align}
The random variable $Z_B$, having distribution known as the Bayes--Chernoff distribution, plays a pivotal role in understanding the limiting coverage. While the distribution of $Z_B$ is analogous to the Chernoff distribution for the MLE, $Z_B$ takes values on the unit interval and exhibits symmetry around $1/2$ ($Z_B = 1 - Z_B$ in distribution). Employing \eqref{eq: one side coverage}, we arrive at the primary conclusion regarding the limiting coverage of $I_{n,\tau}$ as $n \to \infty$:
$$\pb(F_0(x) \in I_{n,\tau}) \to \pb (\tau/2 \le Z_B \le 1-\tau/2).$$
Therefore, the limiting coverage of a $(1-\tau)$-credible interval does not necessarily equal $(1-\tau)$. It can, however, be computed using the distribution of $Z_B$. Let us define $A(u) = \pb(Z_B \le u)$ for $u \in [0, 1]$. The function $A$ is continuous, strictly increasing, and maps the interval $[0, 1]$ onto itself. Consequently, the inverse function $A^{-1}$ is well-defined and also strictly increasing. Additionally, the function $A$ features a symmetry property, which is formalized in Lemma 3.5 in \cite{54} and reported here for completeness.
\begin{lemma}[Lemma 3.5 in \cite{54}]\label{lemma: A properties} 
	The random variable $Z_B \in [0,1]$ is symmetrically distributed about $1/2$, and hence $A(1-u) = 1 - A(u)$, for all $u \in [0,1]$, $A^{-1}(1-v) = 1 - A^{-1}(v)$ for all $v \in [0,1]$, and $A(1/2) = 1/2 = A^{-1}(1/2)$.
\end{lemma}
Using these properties, the limiting coverage $\pb(\tau/2 \le Z_B \le 1-\tau/2)$ can be rewritten as $A(1-\tau/2) - A(\tau/2) = A(1-\tau/2) - (1-A(1-\tau/2)) = 2A(1-\tau/2) - 1$. A recalibration technique allows us to achieve any desired asymptotic coverage, say $1-\beta$. For a two-sided credible interval $I_{n,\tau}$, Theorem \ref{thm: bayesian UQ} and Lemma \ref{lemma: A properties} imply that its asymptotic coverage is $2A(1-\tau/2)-1$. To obtain a desired coverage of $1-\beta$, we set $2A(1-\tau/2)-1 = 1-\beta$. This implies $1-\tau/2 = A^{-1}(1-\beta/2)$, or $\tau/2 = 1 - A^{-1}(1-\beta/2) = A^{-1}(\beta/2)$. Thus, we choose $\tau = 2A^{-1}(\beta/2)$ (c.f.\ Corollary 3.6 in \cite{54}). 
\begin{corollary}\label{cor: coverage}
	For any $0 < \beta < 1$, $\pb (F_0(x) \in I_{n,2A^{-1}(\beta/2)}) \to 1-\beta$.
\end{corollary}

\begin{figure}[H]
    \centering
    \includegraphics[width=0.9\textwidth]{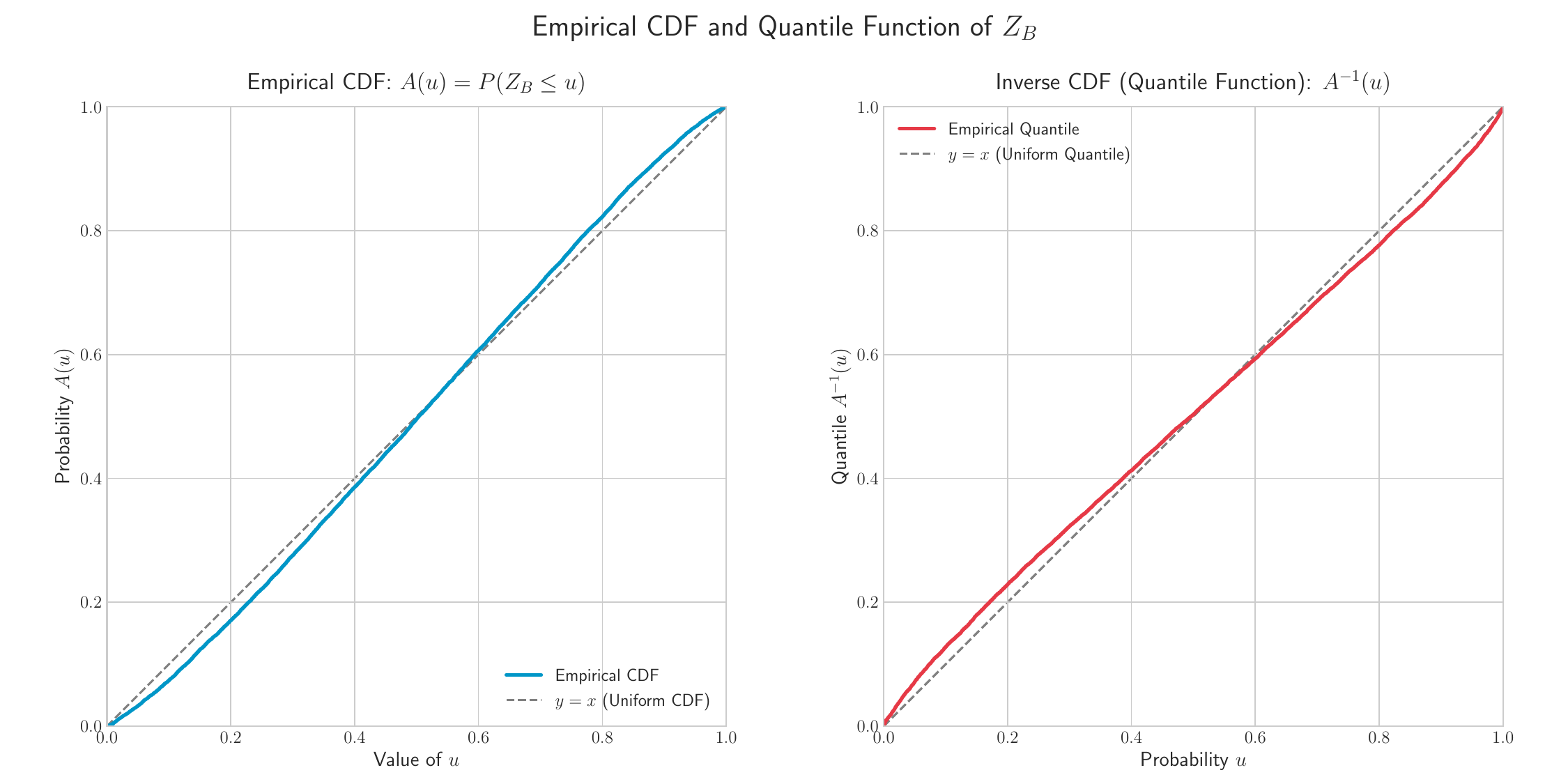}
    \caption{The Bayes--Chernoff distribution $Z_B$ and its inverse function $A^{-1}$, with the identity line for reference (the computations are based on 20000 gathered samples of $Z_B$).}
	\label{img:ZB_cdf_and_inverse}
\end{figure}
\vspace{-0.5cm}
The back-calculation of the required credibility level $1-\tau = 1-2A^{-1}(\beta/2)$ to achieve a coverage level $1-\beta$ can be obtained from a Table 2 in \cite{54}, which presents values of the function $A^{-1}$. For convenience of the reader and for completeness, we report Table 2 from \cite{54} here below, in Table \ref{tab:A_inverse_values} (we verified via our simulations its correctness). Numerical calculations suggest that $A(u) \ge u$ for all values of $u \ge 1/2$ where the computation has been numerically performed and the results are shown in Fig.\ \ref{img:ZB_cdf_and_inverse}. This implies that a credible interval typically has asymptotic coverage greater than its nominal coverage, which may be called a reverse Cox phenomenon.
\vspace{-0.5cm}
\begin{table}[H]
\centering
\caption{Table of the values of the function $A^{-1}$}\label{tab:A_inverse_values}
\begin{adjustbox}{max width=\textwidth}
\begin{tabular}{lccccccccccc}
\hline
$v$ & 0.700 & 0.750 & 0.800 & 0.850 & 0.900 & 0.910 & 0.920 & 0.930 & 0.940 & 0.950 & 0.960 \\
\vspace{0.1cm}
$A^{-1}(v)$ & 0.677 & 0.723 & 0.771 & 0.820 & 0.874 & 0.885 & 0.897 & 0.909 & 0.922 & 0.934 & 0.946 \\
$v$ & 0.965 & 0.970 & 0.975 & 0.980 & 0.985 & 0.990 & 0.995 & 0.996 & 0.997 & 0.998 & 0.999 \\
$A^{-1}(v)$ & 0.952 & 0.960 & 0.966 & 0.973 & 0.980 & 0.986 & 0.993 & 0.995 & 0.996 & 0.997 & 0.999 \\
\hline
\end{tabular}
\end{adjustbox}
\end{table}

\section{Simulation study}

In this section, we present simulations to illustrate the practical behavior of the IIP in comparison to the IIE studied in \cite{74}. In this simulation study, we analyze the setting where the underlying cdf \( F_0 \) is exponential with rate 1.2. Samples \( Z_1, \ldots, Z_n \) used to compute the projection posterior of the IIP are generated hierarchically: \( X_1, \ldots, X_n \sim F_0 \), \( Y_1, \ldots, Y_n \sim k\), with \( Z_i = X_i + Y_i \) for all \( i \). The noise density $k$ which crucially determines the solvability of the Volterra integral equation in \eqref{eq:1}, is chosen from the following set: $\{ \mathrm{Exp}(1), \mathrm{Mix}: \frac{1}{4} \mathrm{Erlang}(2)- \frac{3}{4} \mathrm{Exp}(2), \mathrm{HalfNorm}(0,1), \mathrm{HalfCauchy}(0,2), \mathrm{Lomax}(10,1), \mathrm{HalfLogistic}\}$. We recall the associated kernels $k$ in Fig.\ \ref{img:resolvent_kernels_grid} and we plot the numerical solutions $p$ (and the analytical solutions when available) of the Volterra integral equation in \eqref{eq:1} for each of these noise distributions. 
\begin{figure}[H]
    \centering
    \includegraphics[width=\textwidth]{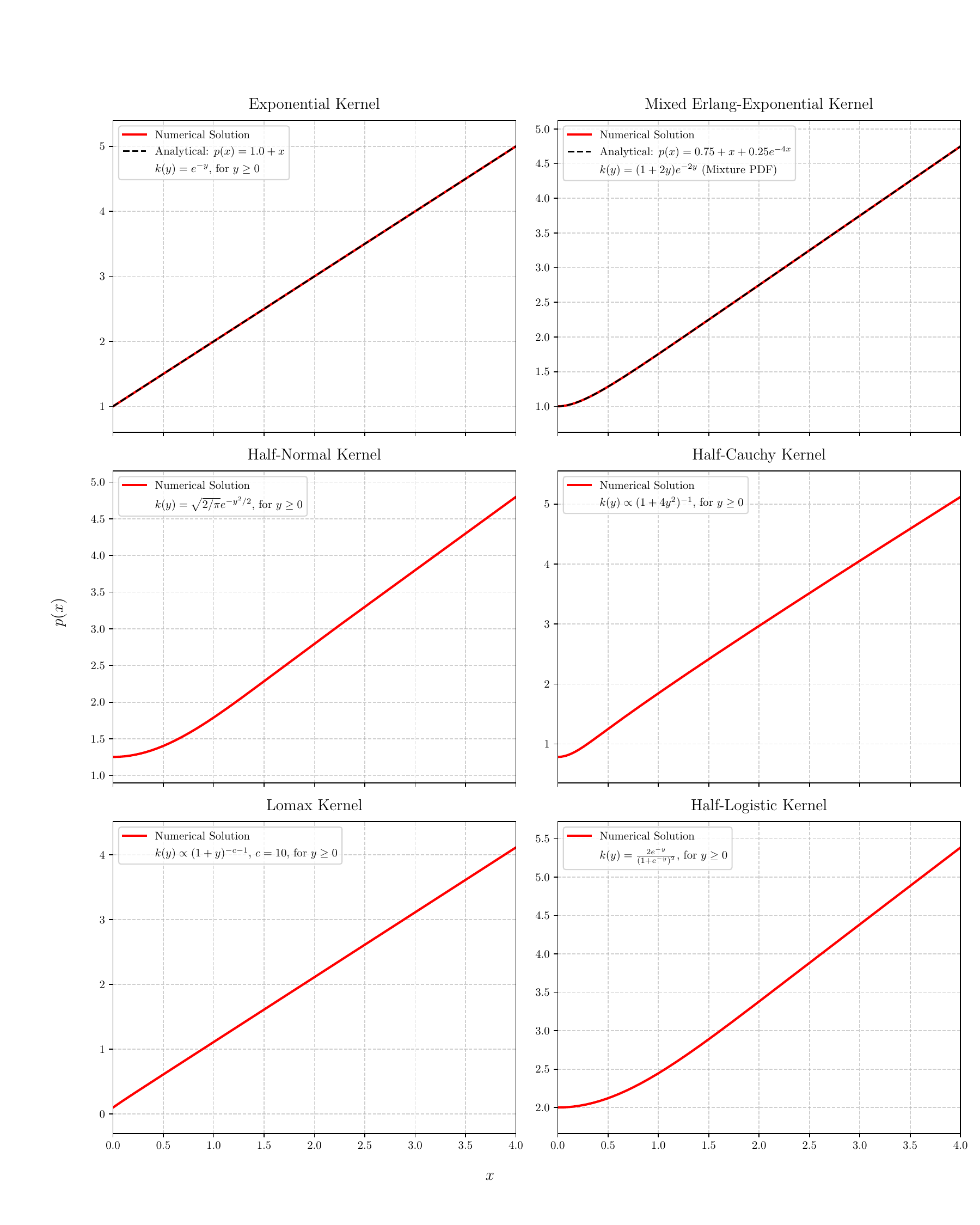}
    \caption{Resolvent of proposed kernels.}
	\label{img:resolvent_kernels_grid}
\end{figure}

\begin{figure}[H]
    \centering
    \includegraphics[width=\textwidth]{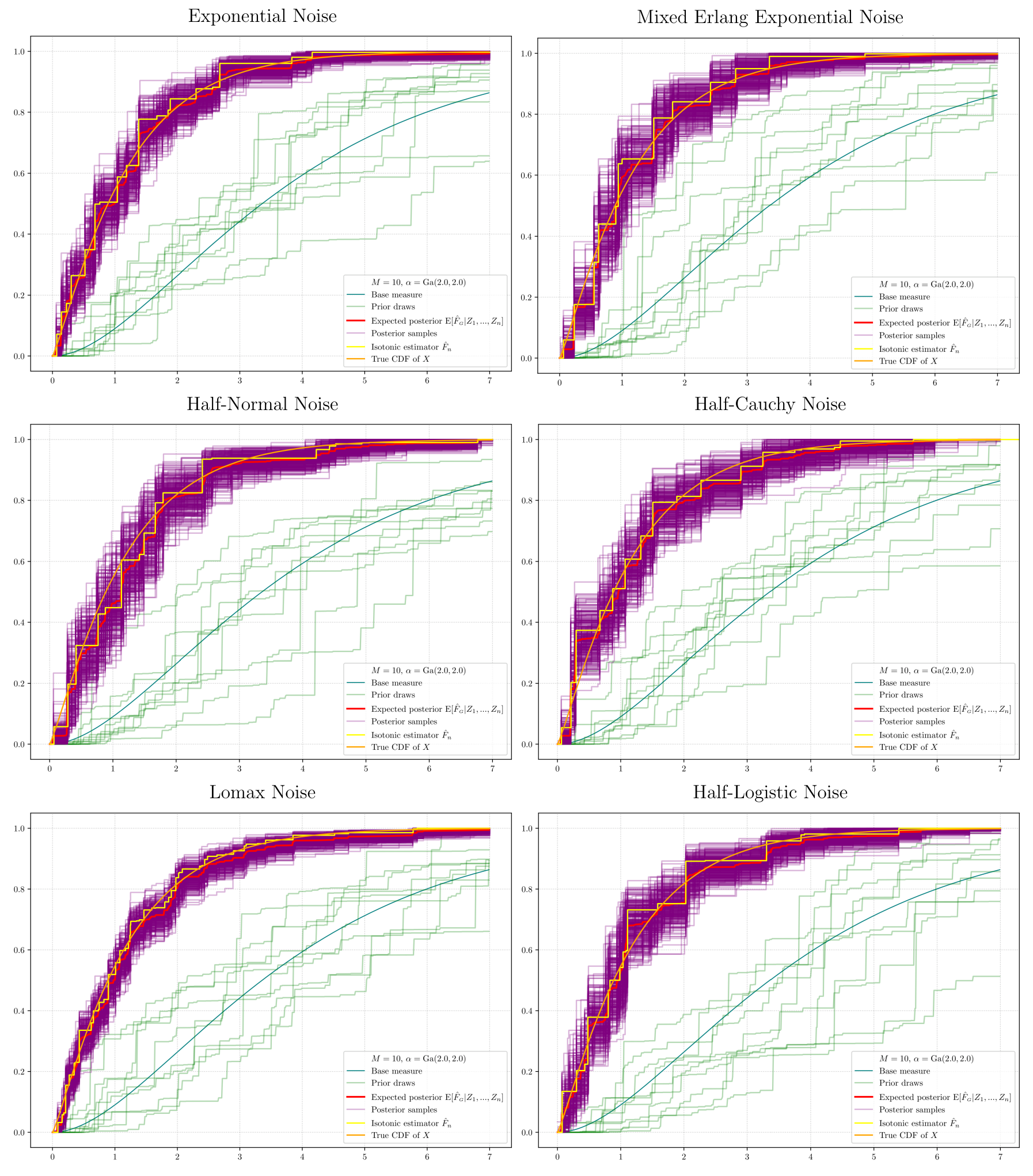}
    \caption{Combined deconvolution plots for $n=200$. In green, we plot draws of \( \hat{F}_{\scriptscriptstyle{G}} \) based on prior draws from \( G \sim \operatorname{DP}(M \alpha) \) with $\alpha \sim \mathrm{Ga}(2,2)$ and prior precision $M =10$; in purple, draws from the isotonized posterior; in red, the average of these draws (approximating the posterior mean); in yellow, the (IIE); and in orange \( F_0 \).}
	\label{img:deconvolution_plots}
\end{figure}

\begin{figure}[H]
    \centering
    \includegraphics[width=\textwidth]{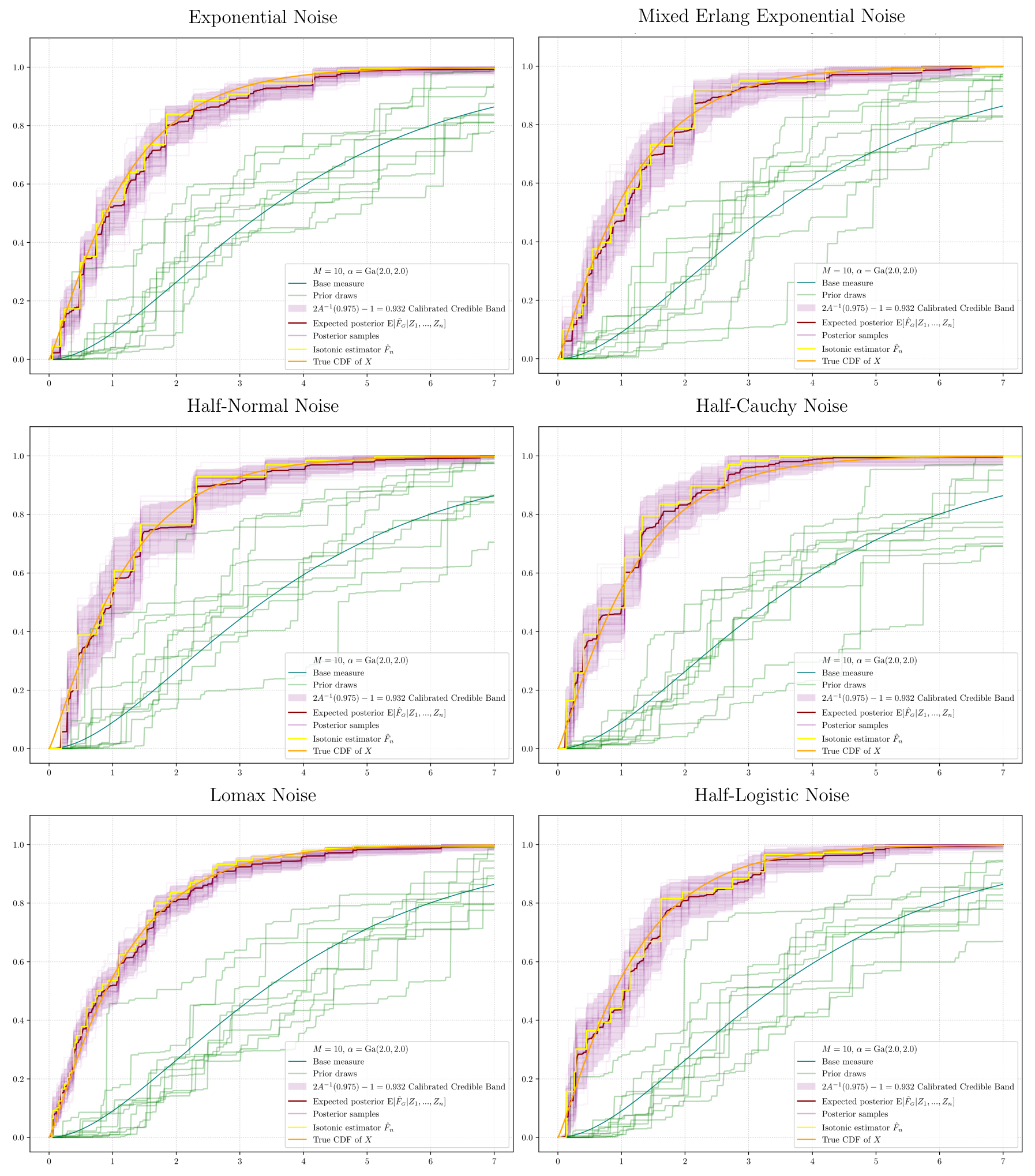}
    \caption{Combined deconvolution plots with calibrated uncertainty quantification for $n=200$. In green, we plot draws of \( \hat{F}_{\scriptscriptstyle{G}} \) based on prior draws from \( G \sim \operatorname{DP}(M \alpha) \) with $\alpha \sim \mathrm{Ga}(2,2)$ and prior precision $M =10$; in pink, the calibrated credible bands from the isotonized posterior; in red, the average of these draws (approximating the posterior mean); in yellow, the (IIE); and in orange \( F_0 \).}
	\label{img:UQ}
\end{figure}

The reasons why we chose this set of distributions are two: on the one hand, for our methodology to work we need $k(0) \neq 0$, with $k$ living on the positive real line. On the other hand, we want to illustrate the performance of the IIP in a variety of settings, including both light-tailed and heavy-tailed noise distributions. Thus we included a good variety of distributions, that go from light-tailed (e.g.\ exponential and half-normal) to heavy-tailed (e.g.\ Lomax with $c=10$, with tails proportional to $y^{-11}$ and which is a special case of Pareto) to very heavy-tailed (e.g.\ half-Cauchy with tails proportional to $y^{-2}$). Because the underlying cdf $F_0$ is always the same, our simulations show that the IIP is able to recover the underlying cdf $F_0$ in all these settings, even when the noise distribution is very heavy-tailed. In particular, in Figure \ref{img:deconvolution_plots} we show the results of the IIP deconvolution methodology compared to the IIE, visualizing the isonotized posterior draws, for the same kernels of Figure \ref{img:resolvent_kernels_grid}. We observe that the Bayesian methodology agrees with the frequentist one. In Figure \ref{img:UQ}, we provide in the same simulation setting the recalibrated credible sets according to the recalibration $2A^{-1}(0.975) -1 = 0.932$. The numerical solutions of the Volterra integral equation are computed using the trapezoidal rule, as described in detail here below. The plots use the following color scheme: green for draws of \( \hat{F}_{\scriptscriptstyle{G}} \) based on prior draws from \( G \sim \operatorname{DP}(M \alpha) \) (parameters \( M \) and \( \alpha \) a pdf specified in the plots not necessarily in the above mentioned set); purple for draws from the isotonized posterior; red for the average of these draws (approximating the posterior mean);  yellow for the Isotonic Inverse Estimator (IIE); and orange for \( F_0 \). 

Overall, these simulations demonstrate that the IIP reliably recovers the underlying distribution $F_0$ across a wide spectrum of noise distributions, ranging from light-tailed to extremely heavy-tailed. Despite substantial variation in the shape and tail behavior of the kernels $k$, the numerical solutions of the associated Volterra equations remain stable, and the resulting posterior draws closely track the frequentist IIE benchmark. The agreement between methods, together with the well-calibrated credible sets, highlights the robustness of the IIP and its practical effectiveness for deconvolution in diverse settings.

\subsubsection*{Numerical Solver of the Volterra Equation}

In order to compute the estimators in case $p$ is not explicitly known, we use a numerical approximation. The presented simulations need the unknown function $p(x)$ solving the specific Volterra integral equation of the first kind:
\begin{equation}
    x \ind(x) = \int_{0}^{x} k(x-t) p(t) \,dt, \, \, \, \, x \geq 0.
    \label{eq:volterra_specific}
\end{equation}
The kernel $k$ is a known function. The solver uses a numerical method based on the trapezoidal rule. A key requirement for this method is $k(0) \neq 0$.

We begin by discretizing the domain $[0, T]$ into $N$ equally spaced points, $x_n = n h$, $p_n := p(x_n)$ and $k_n := k(x_n)$, for $n=0, 1, \dots, N-1$, with step size $h := T/(N-1)$.  At each grid point $x_n$, Equation \eqref{eq:volterra_specific} becomes:
\begin{equation}
    x_n = nh = \int_{0}^{x_n} k(x_n - t) p(t) \,dt.
\end{equation}
We approximate the integral using the composite trapezoidal rule over the points $t_j = jh$ for $j=0, \dots, n$:
\begin{equation}
    nh \approx h \left[ \frac{1}{2} k(x_n - t_0) p(t_0) + \sum_{j=1}^{n-1} k(x_n - t_j) p(t_j) + \frac{1}{2} k(x_n - t_n) p(t_n) \right].
\end{equation}
Let $k(x_n-t_j) := k((n-j)h) = k_{n-j}$ and $k(x_n-t_n) := k(0) = k_0$, then dividing both sides by the step size $h$:
\begin{equation}
    n = \frac{1}{2} k_n p_0 + \sum_{j=1}^{n-1} p_j k_{n-j} + \frac{1}{2} k_0 p_n.
\end{equation}
We can now rearrange this equation to solve for the unknown term $p_n$, assuming all previous values ($p_0, \dots, p_{n-1}$) are known. This gives the recurrence relation for $p_n$:
\begin{equation}
    p_n = \frac{2}{k_0} \left[ n - \frac{1}{2} k_n p_0 - \sum_{j=1}^{n-1} p_j k_{n-j} \right].
    \label{eq:recurrence_specific}
\end{equation}

To start the recurrence, we need the value of $p_0 = p(0)$. This can be found by differentiating the original integral equation using the Leibniz integral rule, which yields the relation $1 = k(0)p(0)$.

The algorithm implemented in the simulation proceeds as follows:
\begin{enumerate}
    \item Discretize the domain $[0, T]$ into $N$ points with step size $h$. Pre-compute the kernel values $k_n$ for $n=0, \dots, N-1$.
    \item Calculate the exact initial value $p_0$: $p_0 = 1/k_0$.
    \item For each step $n = 1, 2, \dots, N-1$, sequentially compute $p_n$ using the recurrence relation from Equation \eqref{eq:recurrence_specific}. 
\end{enumerate}
This procedure yields a stable solution for all examples in Figure \ref{img:resolvent_kernels_grid}.

\section{Conclusion}

This work introduces and evaluates a nonparametric Bayesian method for uncertainty quantification in a specific class of statistical deconvolution problems. To address the challenge of nuisance-parameter-dependent confidence intervals in frequentist methods, we propose the "Isotonic Inverse Posterior." This approach combines a Dirichlet Process prior on the distribution of the observed data with an inverse operator and an isotonic regression step to produce valid estimates for the signal's cumulative distribution function, \(F_0\). The procedure is computationally fast and, following a straightforward recalibration, its credible sets achieve asymptotically correct frequentist coverage without relying on nuisance parameter estimation. As demonstrated through simulations, the method performs effectively across a variety of noise distributions.

\subsection*{Acknowledgments}

The first author gratefully acknowledges the financial support of Aad van der Vaart through his NWO Spinoza Prize. We also thank Aad van der Vaart for valuable feedback on an earlier version of this manuscript.

\newpage
\begin{appendix}
\section{Complementary results}
\textit{Measure-theoretical framework}: in this paper, the sample space is $(\mathbb{R}^+, \mathscr{B}(\mathbb{R}^+))$, and $\mathfrak{M}$ (with Borel $\sigma$-field $\mathscr{M}$ for the weak topology) is the collection of all Borel probability measures on $(\mathbb{R}^+, \mathscr{B}(\mathbb{R}^+))$. The Dirichlet process prior is a probability measure on $\mathfrak{M}$. A prior on the set of probability measures $\mathfrak{M}$ is a probability measure on a sigma-field $\mathscr{M}$ of subsets of $\mathfrak{M}$. Alternatively, it can be viewed as map from some probability space $(\Omega, \mathscr{U}, \mathrm{Pr})$ into $(\mathfrak{M}, \mathscr{M})$. We can think of the hierarchy of the spaces as follows: $(\Omega, \mathscr{U}, \mathrm{Pr}) \rightarrow (\mathfrak{M}, \mathscr{M}) \rightarrow (\mathbb{R}^+, \mathscr{B}(\mathbb{R}^+))$. We put a Dirichlet process prior over $G$ and write $G \sim \Pi = \mathrm{DP}( \alpha)$ for some base measure $\alpha$ and prior precision $|\alpha|$. This means that for any measurable set $B \in \mathscr{M}$ we have $\Pi(B) = \mathrm{Pr} ( G \in B)$. Throughout the paper, we use the notations $\pb$ and $\ex$, which should be understood as the outer probability measure and outer expectation, respectively, whenever issues of measurability in the considered processes arise (see \cite{5} for a rigorous treatment of these topics). 

Furthermore, we use the notion of conditional convergence in distribution. This happens if the bounded Lipschitz distance between the conditional law of Borel measurable $X_n$ given $\mathcal{B}_n$ and a Borel probability measure $L$ tends to zero, where the convergence can be in outer probability. The sequence $X_n$ tends to $L$ conditionally given $\mathcal{B}_n$ in outer probability if
\begin{align}\label{eq: formal definition conditional convergence}
d_{\scriptscriptstyle{ B L}}\left(\mathcal{L}\left(X_n \mid \mathcal{B}_n\right), L\right) :=\sup _{f \in B L_1}\left| \, \mathrm{E}\left(f\left(X_n\right) \mid \mathcal{B}_n\right)-\int f d L\right| \xrightarrow{\mathrm{P}} 0,
\end{align}
where $ B L_1$ is the space of bounded Lipschitz functions with Lipschitz constant 1.

In particular if $Z_1,\ldots,Z_n \overset{\text{i.i.d.}}{\sim} \pb_{\scriptscriptstyle{Z}}$, then $X_n \mid Z_1,\ldots,Z_n \rightsquigarrow N(0,\sigma^2)$ means:
\begin{align}\label{eq: formal definition conditional convergence 1}
	d_{\scriptscriptstyle{ B L}}\left( \pb (X_n \in \cdot \mid Z_1,\ldots,Z_n), N(0,\sigma^2) \right) \xrightarrow{\pb_{\scriptscriptstyle{Z}}} 0.
\end{align}

\begin{lemma}\label{lemma: rewriting1}
Let $x \in (0,T)$, $\hat{F}_{\scriptscriptstyle{G}}$ as in \eqref{eq: def IIP} and $F_0$ true underlying cdf. Then for $a \in \mathbb{R}$ and $n$ sufficiently large
\begin{align*}
n^{1/3} (\hat{F}_{\scriptscriptstyle{G}}(x) - F_0(x)) < a \iff \inf \{ t \in [-x n^{1/3}, (T-x) n^{1/3}] : Z_{\scriptscriptstyle{G}}(t) - a t \, \, \text{is minimal} \} > 0,
\end{align*}
where, for $H_{\scriptscriptstyle{G}}$ as in \eqref{eq: naive bayes},
\[
Z_{\scriptscriptstyle{G}}(t) := n^{2/3} \left( H_{\scriptscriptstyle{G}}(x + t n^{-1/3}) - H_{\scriptscriptstyle{G}}(x) - F_0(x) t n^{-1/3} \right).
\]
\end{lemma}
\begin{proof}
Fix $x \in (0,T)$. Then for $a \in \mathbb{R}$ and $n$ sufficiently large
\begin{align*}
n^{1/3} (\hat{F}_{\scriptscriptstyle{G}}(x) - F_0(x)) < a &\iff \hat{F}_{\scriptscriptstyle{G}}(x) < F_0(x) + a n^{-1/3} \\
&\stackrel{\eqref{def: T_G}}{\iff} T_{\scriptscriptstyle{G}}(F_0(x) + a n^{-1/3}) > x \\
&\iff \inf \{ x + t n^{-1/3} \in [0,T] : H_{\scriptscriptstyle{G}}(x + t n^{-1/3}) - H_{\scriptscriptstyle{G}}(x) \\
& \qquad \quad \quad - F_0(x) t n^{-1/3} - a t n^{-2/3} \, \, \text{is minimal} \} > x \\
&\iff \inf \{ t \in [-x n^{1/3}, (T-x) n^{1/3}] : n^{2/3} (H_{\scriptscriptstyle{G}}(x + t n^{-1/3}) \\
& \qquad \quad \quad - H_{\scriptscriptstyle{G}}(x) - F_0(x) t n^{-1/3}) - a t \, \, \text{is minimal} \} > 0.
\end{align*}
\end{proof}

\begin{lemma}\label{lemma: simplifications}
Let $f_t^n(z) := n^{1/3} (p(x+n^{-1/3}t-z)\ind_{[0, x+n^{-1/3}t)}(z) - p(x-z)\ind_{[0, x)}(z) )$
then, for independent $Q \sim \mathrm{DP}(\alpha)$, $\mathbb{B}_n \sim \mathrm{DP}(n \mathbb{G}_n)$ and $V_n \sim \mathrm{Be}(|\alpha|, n)$. As $n \to \infty$, under $G_0$ in probability:
\begin{align*}
\pb (n^{1/3} |V_n (Q-\mathbb{G}_n)f_t^n| > \varepsilon \mid Z_1, \dots, Z_n) & \xrightarrow{G_0} 0. \numberthis \label{eq: V_n term 1}
\end{align*}
and
\begin{align*}
\pb (n^{1/3} |V_n (\mathbb{B}_n-\mathbb{G}_n)f_t^n| > \varepsilon \mid Z_1, \dots, Z_n) \xrightarrow{G_0} 0. \numberthis \label{eq: V_n term 2}
\end{align*}
\end{lemma}
\begin{proof}
	Using the fact that $p$ is Hölder continuous on $(0, \infty)$ of degree $\gamma > 1/2$:
	\begin{align*}
	|f_t^n(z)| &= n^{1/3} |p(x+n^{-1/3}t-z)\ind_{[0, x)}(z) - (x-z)\ind_{[0, x)}(z)| \\
	&\quad \quad + |p(x+n^{-1/3}t-z) \ind_{[x, x+n^{-1/3}t)}(z)| \\
	&\lesssim n^{1/3} |(x+n^{-1/3}t-z) - (x-z)|^\gamma \ind_{[0, x)}(z) \\
	&\quad \quad  + n^{1/3} \sup_{z \in [x, x+n^{-1/3}t]} |p(x+n^{-1/3}t-z)| \\
	&\lesssim n^{1/3} |n^{-1/3}t \ind_{[0, x)}(z)|^\gamma + n^{1/3} \sup_{z \in [x, x+n^{-1/3}t]} |x+n^{-1/3}t-z|^\gamma \\
	&\lesssim n^{1/3} p(0) + n^{(1-\gamma)/3} |t|^\gamma.
	\end{align*}
	Note that $\ex V_n = |\alpha|/(|\alpha|+n)$, thus $\forall \, \varepsilon > 0$ by independence, $\gamma > 1/2$ as $n \rightarrow \infty$
	\begin{align*}
	\pb (n^{1/3} |V_n (Q-\mathbb{G}_n)f_t^n| > \varepsilon \mid Z_1, \dots, Z_n) & \lesssim n^{1/3} \varepsilon^{-1} \ex V_n \, \ex [ (Q+\mathbb{G}_n) |f_t^n| \mid Z_1, \dots, Z_n] \\
	&\lesssim n^{(2-\gamma)/3} \varepsilon^{-1} \frac{|\alpha|}{|\alpha|+n} |t|^\gamma \to 0. 
	\end{align*}
	Similarly $\pb (n^{1/3} |V_n (\mathbb{B}_n-\mathbb{G}_n)f_t^n| > \varepsilon \mid Z_1, \dots, Z_n) \xrightarrow{G_0} 0$.
\end{proof}

\begin{lemma}\label{lemma: further simplification}
For $W_n^*$ as in \eqref{eq: W_n^* simplified}, in probability under $G_0$, as $n \to \infty$:
\begin{align}\label{eq: W_n^* simplified 1}
	W_n^*(t) = \mathbb{W}_n^*(t) + o_p(1) .
\end{align}
where for \( \varepsilon_i \overset{\text{i.i.d.}}{\sim} \mathrm{Exp}(1) \), the process $\mathbb{W}_n^*$ is defined as:
\begin{align}\label{eq: mathbbW star process 1}
    \mathbb{W}_n^*(t) = \frac{1}{\sqrt{n}} \sum_{i=1}^n \frac{(\varepsilon_i-1)}{\sqrt[6]{n}} \tilde{f}_t^n(Z_i).
\end{align}
and $\tilde{f}_t^n(z) := n^{1/3} p(0) \ind_{[x, x+n^{-1/3}t)}(z)$.
\end{lemma}
\begin{proof}
First recall the decomposition:
\begin{align*}
	W_n^*(t) = n^{1/3}(\mathbb{B}_n - \mathbb{G}_n)f_t^n &= \frac{1}{n^{2/3}} \sum_{i=1}^n (W_{ni}-1) f_t^n(Z_i),
\end{align*}
where $W_{ni} = \frac{\varepsilon_i}{\bar{\varepsilon}_n}$. Because $|\bar{\varepsilon}_n - 1| = O_p(1/\sqrt{n})$, for almost every sequence $Z_1, \dots, Z_n$ conclude that uniformly in $t$ in compacta:
\begin{align*}
\frac{1}{n^{2/3}} \sum_{i=1}^n (W_{ni}-1) f_t^n(Z_i) = \frac{1}{\sqrt{n}} \sum_{i=1}^n \frac{(\varepsilon_i-1)}{\sqrt[6]{n}} f_t^n(Z_i) + o_p(1).
\end{align*}
Now consider the decomposition:
\[
f_t^n(z) = \tilde{f}_t^n(z) + r_t^n(z)
\]
where for $\tilde{p} = p - p(0) \ind_{[0,\infty)}$:
\begin{align*}
r_t^n(z) &:= n^{1/3} (\tilde{p}(x + n^{-1/3}t - z) - \tilde{p}(x - z)), \\
\tilde{f}_t^n(z) &:= n^{1/3} p(0) \ind_{[x, x+n^{-1/3}t)}(z).
\end{align*}
We obtain the claim that uniformly in $t$ over compacta:
\begin{align}
	W_n^*(t) = \mathbb{W}_n^*(t) + o_p(1),
\end{align}
as soon as we show that for all compact $K \subset \mathbb{R}$ and $\forall \, \varepsilon > 0$:
\begin{align}\label{eq: representation without remainders}
	\pb \left( \sup_{t \in K} \left| \frac{1}{\sqrt{n}} \sum_{i=1}^n \frac{(\varepsilon_i-1)}{\sqrt[6]{n}} r_t^n(Z_i) \right| > \varepsilon \mid Z_1, \dots, Z_n \right) \xrightarrow{G_0} 0.
\end{align}

For $\tilde{p} = p - p(0)\ind_{[0,\infty)}$, $\tilde{p}$ is $\gamma$-Hölder continuous of
degree $\gamma > 1/2$. Define:
	\[ R_n^*(t) := \frac{1}{n^{1/3}} \sum_{i=1}^n (\varepsilon_i - 1) \left( \tilde{p}(x+n^{-1/3}t-Z_i) - \tilde{p}(x-Z_i) \right). \]
The claim is obtained if we show for any $K \in (0,\infty)$, $\sup_{|t|\le K} |R_n^*(t)| \xrightarrow{G_0} 0$.

Let $0 = t_0 < t_1 < \dots < t_{K_n} = K$ and $t_{-i} = -t_i$, $i=1, \dots, K_n$. Now note:
\begin{align*}
	\sup_{|t|\le K} |R_n^*(t)| &= \max_{-K_n+1 \le i \le K_n} \sup_{t \in [t_{i-1}, t_i]} |R_n^*(t)| \\
	&\le \max_{-K_n+1 \le i \le K_n} \left( |R_n^*(t_i)| + \sup_{t \in [t_{i-1}, t_i]} |R_n^*(t) - R_n^*(t_i)| \right).
\end{align*}
Using Markov's inequality, we have:
\begin{align*}
&\varepsilon \, \pb \left( \sup_{|t|\le K} |R_n^*(t)| > \varepsilon \right) \le \ex \left( \sup_{|t|\le K} |R_n^*(t)| \right) \\
&\quad \quad \le \ex \left( \max_{-K_n+1 \le i \le K_n} |R_n^*(t_i)| \right) + \ex \left( \max_{-K_n+1 \le i \le K_n} \sup_{t \in [t_{i-1}, t_i]} |R_n^*(t) - R_n^*(t_i)| \right).
\end{align*}
For each $t \in [t_{i-1}, t_i]$ and $\mathbb{P}_n = \frac{1}{n} \sum_{i=1}^n \delta_{\varepsilon_i}$
\begin{align*}
&|R_n^*(t) - R_n^*(t_i)| \\
& = n^{2/3} \left| \int_0^\infty \int_0^\infty \left( \tilde{p}(x+n^{-1/3}t-z) - \tilde{p}(x+n^{-1/3}t_i-z) \right) (\varepsilon-1) \, d(\mathbb{G}_n, \mathbb{P}_n)(z,\varepsilon) \right| \\
& \le 2 n^{2/3} L n^{-\gamma/3} |t-t_i|^\gamma \left( 1 + \frac{1}{n} \sum_{j=1}^n \varepsilon_j \right) = 2 n^{(2-\gamma)/3} L |t-t_i|^\gamma \left( 1 + \frac{1}{n} \sum_{j=1}^n \varepsilon_j \right).
\end{align*}
Take a grid of $t_i$'s such that $|t_i - t_{i-1}|^\gamma = \delta n^{-(2-\gamma)/3}$
then
\[ \ex \left( \sup_{-K_n+1 \le i \le K_n} \sup_{t \in [t_{i-1}, t_i]} |R_n^*(t) - R_n^*(t_i)| \right) \le 4 L \delta,\]
(which we can make small by taking $\delta$ small enough; note also that $K_n = O(n^{(2-\gamma)/3\gamma})$). We can rewrite $R_n^*(t) = \sum_{i=1}^n Y_i$ where:
\[ Y_i := n^{-1/3} (\varepsilon_i - 1) \left( \tilde{p}(x+n^{-1/3}t-Z_i) - \tilde{p}(x-Z_i) \right). \]
Note that $Y_i$ has expectation $0$ and:
\[ |Y_i| \le n^{-1/3} (L n^{-\gamma/3} |t|^\gamma + C n^{-1/3} |t|) |\varepsilon_i-1| \]
\[ \le (C' n^{-(1+\gamma)/3} + C K n^{-2/3}) |\varepsilon_i-1| \]
and because the exponential random variables
have finite Orlicz norm $\|\varepsilon_i\|_{\psi_1} < T$, $0 < T < \infty$ (where the Orlicz norm is defined as: $\|X\|_{\psi_1}
:= \inf \{ c > 0 : \mathrm{E} [ \exp ( |X| / c ) ] \le 2 \}$),
then for some finite constant $K'$:
\[ \|Y_i\|_{\psi_1} \le K' (C' n^{-(1+\gamma)/3} + C K n^{-2/3}). \]
Therefore,
\[ \|Y_i\|_{\psi_1}^2 \le 2(K'C')^2 n^{-(2+2\gamma)/3} + 2(KCK'C')^2 n^{-4/3} \le 2(K'C')^2 (1+(KC)^2) n^{-(2+2\gamma)/3}. \]
By Example 2.2.12 and Lemma 2.2.10 in \cite{5}, for some constants $C_{1,2,3}$:
\begin{align*}
\pb (|R_n^*(t)| > x) &= \pb \left(\left|\sum_{i=1}^n Y_i\right| > x\right) \\
&\le 2 \exp\left\{-\frac{1}{2} \frac{x^2}{C_1 n^{(1-2\gamma)/3} + (C_2 n^{-(1+\gamma)/3} + C_3 n^{-2/3}) x}\right\}.
\end{align*}
By applying Lemma 2.2.13 in \cite{5} to $R_n^*(t_{-K_n+1}), \dots, R_n^*(t_{K_n})$
and $\|\cdot\|_1 \lesssim \|\cdot\|_{\psi_1}$:
\begin{align}
\ex \left(\max_{-K_n+1 \le i \le K_n} |R_n^*(t_i)| \right) &\lesssim C \left( (C_2 n^{-(1+\gamma)/3} + C_3 n^{-2/3}) \log(1+2K_n) \right. \\
& \quad \left. + \,  C_1^{1/2} n^{(1-2\gamma)/6} \sqrt{\log(1+2K_n)} \right). 
\end{align}
Because $K_n = O(n^{(2-\gamma)/3\gamma})$, the result follows.
\end{proof}

\begin{lemma}\label{lemma: conditional convergence W star}
For every compact $K \subset \mathbb{R}$ as $n \to \infty$:
\begin{align}\label{eq: marginal convergence}
	\mathbb{W}_n^*\mid Z_1, \dots, Z_n \rightsquigarrow \sqrt{\frac{g_0(x)}{k(0)}} \mathbb{W}_2 \quad \text{in } \, \, \ell^\infty(K).
\end{align}
where $\mathbb{W}_n^*$ is as in \eqref{eq: mathbbW star process 1} and $\mathbb{W}_2$ is a two sided Brownian Motion originating in $0$.
\end{lemma}
\begin{proof}
First, for $t \in \mathbb{R}$, we give the proof of the marginal convergence 
\begin{align}\label{eq: first marginal convergence}
    \mathbb{W}_n^*(t) \mid Z_1, \ldots, Z_n \rightsquigarrow \mathbb{W}_2(t)
\end{align}
where $\mathbb{W}_2$ is a two sided Brownian Motion originating in $0$. After that, we will show conditional asymptotic equicontinuity of the process $\mathbb{W}_n^*$, proving the claim \eqref{eq: marginal convergence}.

Without loss of generality, let $t \ne 0$ in $\mathbb{R}$. We apply Theorem 6.16 in \cite{53} with	$\xi_{ni}$ equal to:
\[ \xi_{ni} := \frac{(\varepsilon_i - 1) \tilde{f}_t^n (Z_i)}{n^{2/3}} = \frac{(\varepsilon_i - 1)}{n^{1/3}} p(0) \ind_{[x, x+n^{-1/3}t]}(Z_i). \]
Therefore we need to show $\forall \, \delta > 0$:
\begin{align*}
&\sum_{i=1}^n \pb (|\varepsilon_i - 1| |\tilde{f}_t^n(Z_i)| > \delta n^{2/3} \mid Z_1, \dots, Z_n) \xrightarrow{G_0} 0, \numberthis \label{eq: first} \\
&\frac{1}{n^{4/3}} \sum_{i=1}^n \ex \left( (\varepsilon_i - 1) \tilde{f}_t^n(Z_i) \ind_{\{|\varepsilon_i-1||\tilde{f}_t^n(Z_i)| \ge \delta n^{2/3}\}} \mid Z_1, \dots, Z_n \right) \xrightarrow{G_0} 0, \quad \text{and} \numberthis \label{eq: second} \\
& \frac{1}{n^{4/3}} \sum_{i=1}^n \text{Var} \left( (\varepsilon_i - 1) \tilde{f}_t^n(Z_i) \ind_{\{|\varepsilon_i-1||\tilde{f}_t^n(Z_i)| < \delta n^{2/3}\}} \mid Z_1, \dots, Z_n \right) \xrightarrow{G_0} \frac{g_0(x)}{k(0)^2} t. \numberthis \label{eq: third} 
\end{align*}
For \eqref{eq: first}, use that $|\tilde{f}_t^n(Z_i)| \le p(0) n^{1/3}$, for all $n$ big enough, unconditionally:
\begin{align*} 
	\sum_{i=1}^n \pb (|\varepsilon_i - 1| |\tilde{f}_t^n(Z_i)| > \delta n^{2/3}) &\leq n \, \pb (|\varepsilon_i - 1| > \delta n^{1/3} / p(0)) \\ &= n \int_{\lambda > 1 + |t|^{-1} \delta n^{1/3}} e^{-\lambda} d\lambda = n e^{-1-|t|^{-1}\delta n^{1/3}} \to 0. 
\end{align*}
Now consider \eqref{eq: second} and recall $Z \sim g_0$ and $\varepsilon \sim \mathrm{Exp}(1)$ in what follows. Because $\ex \tilde{f}_t^n(Z) = O(1)$ since, as $n \to \infty$:
\begin{align*}
	\ex \tilde{f}_t^n(Z) &= n^{1/3} \int p(0) \ind_{[x, x+n^{-1/3}t]}(z) g_0(z) \, dz\\
	&= p(0) n^{1/3} \, \pb (Z \in [x, x+n^{-1/3}t]) \to t p(0) g_0(x),
\end{align*}
and $\varepsilon \indep Z$, we conclude that:
\begin{align*} 
	&\ex \left| \frac{1}{n^{2/3}} \sum_{i=1}^n \ex \left( (\varepsilon_i - 1) \tilde{f}_t^n(Z_i) \ind_{\{|\varepsilon_i-1||\tilde{f}_t^n(Z_i)| \ge \delta n^{2/3}\}} \mid Z_1, \dots, Z_n \right) \right| \\ 
	&\lesssim n^{1/3} \ex (|\varepsilon-1| \ind_{\{|\varepsilon-1| \ge \delta n^{1/3} / p(0)\}}) \ex \tilde{f}_t^n(Z) \le n^{1/3} (1+\delta n^{1/3}) e^{-1-\delta n^{1/3}} \to 0.
\end{align*}

For \eqref{eq: third}, some more work is needed. Note that:
\begin{align*}
&\eqref{eq: third} = \frac{1}{n^{4/3}} \sum_{i=1}^n \bigg\{ \ex \left( (\tilde{f}_t^n(Z_i))^2 (\varepsilon_i-1)^2 \ind_{\{|\varepsilon_i-1|\tilde{f}_t^n(Z_i) < \delta n^{2/3}\}} \mid Z_1, \dots, Z_n \right) \\
&\quad \quad \quad \quad \quad  \quad \quad \quad  - \left( \ex \left( \tilde{f}_t^n(Z_i) (\varepsilon_i-1) \ind_{\{|\varepsilon_i-1|\tilde{f}_t^n(Z_i) < \delta n^{2/3}\}} \mid Z_1, \dots, Z_n \right) \right)^2 \bigg\}.
\end{align*} 
First notice, because $\ex (\varepsilon-1)=0$ and $\varepsilon \indep Z$:
\begin{align*}
&\ex \left( \frac{1}{n^{4/3}} \sum_{i=1}^n \left( \ex \left( \tilde{f}_t^n(Z_i) (\varepsilon_i-1) \ind_{\{|\varepsilon_i-1|\tilde{f}_t^n(Z_i) > \delta n^{2/3}\}} \mid Z_1, \dots, Z_n \right) \right)^2 \right) \\
&\quad \le \frac{1}{n^{1/3}} \ex \left( (\tilde{f}_t^n(Z))^2 (\varepsilon-1)^2 \ind_{\{|\varepsilon-1|\tilde{f}_t^n(Z) > \delta n^{2/3}\}} \right) \\
&\quad  \leq \frac{1}{n^{1/3}} \ex \left( (\tilde{f}_t^n(Z))^2 (\varepsilon-1)^2 \ind_{\{|\varepsilon-1| > \delta n^{1/3}/ p(0)\}} \right) \to 0,
\end{align*}
because as $n \rightarrow \infty$
\[ \frac{1}{n^{1/3}} \ex \left( (\tilde{f}_t^n(Z))^2 \right) = n^{1/3} \int p^2(0) \ind_{[x, x+n^{-1/3}t]}(z) g_0(z) dz \to p^2(0) g_0(x) t, \]
and for all $n$ big enough:
\begin{align*} 
	\ex \left( (\varepsilon-1)^2 \ind_{\{|\varepsilon-1| > \delta n^{1/3}\}} \right) &\lesssim \int_{1+\delta n^{1/3}}^\infty (\lambda-1)^2 e^{-\lambda} d\lambda \\ &= (\delta n^{1/3} (\delta n^{1/3} + 2) + 2) e^{-1-\delta n^{1/3}} \to 0. \numberthis \label{eq: tail exponentials}
\end{align*}
Now using that $\ex(\varepsilon-1)^2 = 1$ and that $\varepsilon \indep Z$, we get for the domination term of \eqref{eq: third}:
\begin{align*} 
	&\frac{1}{n^{4/3}} \sum_{i=1}^n \ex \left( (\tilde{f}_t^n(Z_i))^2 (\varepsilon_i-1)^2 \ind_{\{|\varepsilon_i-1|\tilde{f}_t^n(Z_i) < \delta n^{2/3}\}} \mid Z_1, \dots, Z_n \right) \\ 
	&= \frac{1}{n^{4/3}} \sum_{i=1}^n \ex \left( (\tilde{f}_t^n(Z_i))^2 (\varepsilon_i-1)^2 (1 - \ind_{\{|\varepsilon_i-1|\tilde{f}_t^n(Z_i) \ge \delta n^{2/3}\}}) \mid Z_1, \dots, Z_n \right) \\ 
	&= \frac{1}{n^{4/3}} \sum_{i=1}^n (\tilde{f}_t^n(Z_i))^2 - \frac{1}{n^{4/3}} \sum_{i=1}^n \ex \left( (\tilde{f}_t^n(Z_i))^2 (\varepsilon_i-1)^2 \ind_{\{|\varepsilon_i-1|\tilde{f}_t^n(Z_i) \ge \delta n^{2/3}\}} \mid Z_1, \dots, Z_n \right).
\end{align*}
Now we show $\frac{1}{n^{4/3}} \sum_{i=1}^n (\tilde{f}_t^n(Z_i))^2 \xrightarrow{G_0} t \frac{g_0(x)}{k^2(0)}$. We apply Lemma 2.3 from the supplement of \cite{65} with $\xi_{ni} := (\tilde{f}_t^n(Z_i))^2 / n^{4/3}$. For all $\varepsilon > 0$ and all $n$ big enough:
\begin{align*} 
	\sum_{i=1}^n \ex |\xi_{ni}| \ind_{\{|\xi_{ni}| > \varepsilon\}} &= \frac{1}{n^{4/3}} \sum_{i=1}^n \ex \left( (\tilde{f}_t^n(Z_i))^2 \ind_{\{(\tilde{f}_t^n(Z_i))^2 > \varepsilon n^{4/3}\}} \right) \\ 
	&\le \frac{1}{n^{4/3}} \sum_{i=1}^n \ex \left( (\tilde{f}_t^n(Z_i))^2 \ind_{\{p^2(0) > \varepsilon n^{2/3}\}} \right) = 0. 
\end{align*}
Finally note, because $p^2(0) = \frac{1}{k^2(0)}$:
\[ \sum_{i=1}^n \ex \xi_{ni} = n^{1/3} p^2(0) \int \ind_{[x, x+n^{-1/3}t]}(z) g_0(z) \, dz \to t \frac{g_0(x)}{k^2(0)}. \]
The fact that the $\xi_{ni}$ are nonnegative and $\sum_{i=1}^n \ex \xi_{ni} \to t \frac{g_0(x)}{k^2(0)}$
imply $\sum_{i=1}^n \ex |\xi_{ni}| = O(1)$. Using \eqref{eq: tail exponentials} we conclude:
\[ \frac{1}{n^{4/3}} \sum_{i=1}^n \ex \left( (\varepsilon_i-1)^2 (\tilde{f}_t^n(Z_i))^2 \ind_{\{|\varepsilon_i-1|\tilde{f}_t^n(Z_i) > \delta n^{2/3}\}} \right) \]
\[ \le \frac{1}{n^{4/3}} \sum_{i=1}^n (\tilde{f}_t^n(Z_i))^2 \ex \left( (\varepsilon-1)^2 \ind_{\{|\varepsilon-1| > \delta n^{1/3}\}} \right) \to 0. \]
because $\frac{1}{n^{4/3}} \sum_{i=1}^n (\tilde{f}_t^n(Z_i))^2 = O_{\scriptscriptstyle{G_0}}(1)$.

Now we verify the behavior of the covariance structure. Note that 
$\forall s, t \geq 0$ we have the following convergence in probability as $n \rightarrow \infty$: $\mathrm{Cov}(\mathbb{W}_n^*(t), \mathbb{W}_n^*(s) \mid Z_1, \dots, Z_n)$ converges in probability under $G_0$ as $n \rightarrow \infty$ to:
\begin{align*} 
	\mathrm{Cov}(\mathbb{W}_n^*(t), \mathbb{W}_n^*(s)) &= n^{1/3} p^2(0) \int (\ind_{[x, x+n^{-1/3}t]} - \ind_{[0,x]})(z) (\ind_{[x, x+n^{-1/3}s]} - \ind_{[0,x]})(z) \, dG_0(z) \\ &= n^{1/3} p^2(0) \, \pb (Z \in [x, x+n^{-1/3} \min(s,t)]) \to  \frac{g_0(x)}{k^2(0)} \min(s,t),
\end{align*}
conclude the convergence marginal convergence in \eqref{eq: first marginal convergence}.

Now to prove
\begin{align}
\mathbb{W}_n^*\mid Z_1, \dots, Z_n \rightsquigarrow \sqrt{\frac{g_0(x)}{k(0)}} \mathbb{W}_2 \quad \text{in } \, \, \ell^\infty(K)
\end{align}
take the random class of functions (for $\varepsilon \sim \mathrm{Exp}(1)$), for $T >0$
\[
\tilde{\mathcal{F}}_n^{T,\delta} = \{ z \mapsto n^{-1/6}(\varepsilon-1)(\tilde{f}_{t}^n - \tilde{f}_{s}^n)(z) : s, t \in \mathbb{R}, |s-t| < \delta, \, \max\{|s|,|t|\} < T \},
\]
and the class of functions $\mathcal{H}_{n,3}^{T,\delta}$ whose functions for $f \in \tilde{\mathcal{F}}_n^{T,\delta}$ are given by $f/(\varepsilon-1)$ thus without the $\varepsilon-1$. These classes have square integrable envelopes $\tilde{F}_{n,\delta}$, $H_{n,\delta}$ that satisfy for all $\delta > 0$ small and all $n$ big enough:
\begin{align*}
\| \tilde{F}_{n,\delta} \|_2^2 &= \ex (\varepsilon-1)^2 \|H_{n,\delta}\|_{2,G_0}^2 = \frac{1}{n^{1/3}} \int (\tilde{f}_{t}^n - \tilde{f}_{s}^n)^2(z) g_0(z) \, dz\\
&= n^{1/3} p^2(0) \int \ind_{[x+n^{-1/3}s \wedge t, x+n^{-1/3}s \vee t]}(z) g_0(z) \, dz \lesssim p^2(0) g_0(x) \delta.
\end{align*}
For fixed $s,t \in \mathbb{R}$, in view of \eqref{eq: W_n^* simplified}, we prove conditional asymptotic equicontinuity of $\mathbb{W}_n^*$, i.e. $\forall \, \eta > 0$ as $n \to \infty$, followed by $\delta > 0$ the right hand side of the expression below converges to $0$ under $G_0$:
\[
\pb \left( \sup_{|s-t|<\delta} |\mathbb{W}_n^*(t) - \mathbb{W}_n^*(s)| > \eta \mid Z_1, \dots, Z_n \right) \leq \frac{1}{\eta} \ex \left[ \sup_{|s-t|<\delta} |\mathbb{W}_n^*(t) -\mathbb{W}_n^*(s)| \mid Z_1, \dots, Z_n \right].
\]
By the maximal inequality in Theorem 2.14.1 in \cite{5}, as $\delta \rightarrow 0$
\[
\ex \left[ \sup_{|s-t|<\delta} |\mathbb{W}_n^*(t) - \mathbb{W}_n^*(s)| \right] \lesssim J(1, \tilde{\mathcal{F}}_n^{T,\delta}, L_2) \| \tilde{F}_{n,\delta} \|_2 \lesssim \sqrt{\delta} \to 0.
\]
In the last convergence we use that:
\[
J(1, \tilde{\mathcal{F}}_n^{T,\delta}, \mathbb{L}_2) := \sup_R \int_0^1 \sqrt{\log N(\eta \| \tilde{F}_{n,\delta} \|_2, \tilde{\mathcal{F}}_n^{T,\delta}, \mathbb{L}_2(R))} \, d\eta,
\]
where the $\mathbb{L}_2(R)$ - square distance is given for $s,t$ and $s', t' \in \mathbb{R}$ by:
$$\frac{1}{n^{1/3}} \int \left( (\tilde{f}_{s}^n - \tilde{f}_{t}^n)(z) - (\tilde{f}_{s'}^n - \tilde{f}_{t'}^n)(z) \right)^2 dR(z), $$
(note that we did not take $\varepsilon -1 $ into account as $\ex(\varepsilon -1 )^2=1$). 
Furthermore we use $J(1, \tilde{\mathcal{F}}_n^{T,\delta}, \mathbb{L}_2) < \infty$ as a consequence of Exercise 20 in \cite{5} (p.\ 221), which implies that the function classes under consideration are VC of uniform bounded index.

All together this proves \eqref{eq: marginal convergence}.

\end{proof}

\begin{lemma}\label{lemma: consistency}
	Assume Condition \ref{cond:1} holds and that $F_0$ is continuous at $x$. For any $T > 0$ and $x \in [0, T)$, as $n \rightarrow \infty$, $\forall \, \varepsilon > 0$:
	\[ \Pi_n ( |\hat{F}_{\scriptscriptstyle{G}}(x) - F_0(x)| > \varepsilon \mid Z_1, \dots, Z_n ) \xrightarrow{G_0} 0. \]
\end{lemma}
\begin{proof}
Using Proposition 4.3 in \cite{4}, we obtain:
\begin{align*} 
	&\ex (H_{\scriptscriptstyle{G}}(x) \mid Z_1, \dots, Z_n) = \frac{1}{|\alpha|+n} \int_0^{x} p(x-z) \, d \, \left(\alpha + \sum_{i=1}^n \delta_{Z_i}\right)(z) \\
	&\mathrm{Var}(H_{\scriptscriptstyle{G}} (x) \mid Z_1, \dots, Z_n) = \frac{1}{1+|\alpha|+n} \frac{1}{|\alpha|+n} \int_0^{x} p^2(x-z) \, d\left(\alpha + \sum_{i=1}^n \delta_{Z_i}\right)(z) 
\end{align*}
Thus, using the conditional Chebyshev inequality, $\forall \, \varepsilon > 0$:
\[ \Pi_n ( |H_{\scriptscriptstyle{G}}(x) - H_0(x)| > \varepsilon \mid Z_1, \dots, Z_n ) \stackrel{G_0}{\to} 0 \quad \text{a.s.} \]
Because of the continuity of $H_0$, the monotonicity of $H_{\scriptscriptstyle{G}}$ and $H_0$ and that $H_{\scriptscriptstyle{G}}(0) = H_0(0) = 0$, $\forall \, T > 0, \varepsilon > 0$:
\[ \Pi_n \left( \sup_{0 \leq x \leq T} |H_{\scriptscriptstyle{G}} (x) - H_0(x)| > \varepsilon \mid Z_1, \dots, Z_n \right) \to 0 \quad \text{a.s.} \]
Now, since $x \in [0, T)$ and $F_0$ is continuous at $x$, fix $\varepsilon > 0$ and $h > 0$ such that:
\[ \left| \frac{H_0(x) - H_0(x-h)}{h} - F_0(x) \right| < \varepsilon \]
By the convexity of $H_{\scriptscriptstyle{G}}^{\star}$ (recall that it is the GCM of $H_{\scriptscriptstyle{G}}$) we have:
\[ \frac{H_{\scriptscriptstyle{G}}^{\star}(x) - H_{\scriptscriptstyle{G}}^{\star}(x-h)}{h} \le \hat{F}_{\scriptscriptstyle{G}}(x) \le \frac{H_{\scriptscriptstyle{G}}^{\star}(x+h) - H_{\scriptscriptstyle{G}}^{\star}(x)}{h} \]
Thus:
\begin{align*}
&\Pi_n ( \hat{F}_{\scriptscriptstyle{G}} (x) - F_0(x) > 2\varepsilon \mid Z_1, \dots, Z_n ) \\
&\quad \quad \quad \quad  \quad \le \Pi_n \left( \frac{H_{\scriptscriptstyle{G}}^{\star}(x) - H_{\scriptscriptstyle{G}}^{\star}(x-h)}{h} \ge \frac{H_0(x) - H_0(x-h)}{h} + \varepsilon \mid Z_1, \dots, Z_n \right) \xrightarrow{G_0} 0. 
\end{align*}
Similarly:
\[ \Pi_n ( \hat{F}_{\scriptscriptstyle{G}}(x) - F_0(x) < -2\varepsilon \mid Z_1, \dots, Z_n ) \xrightarrow{G_0} 0. \]
\end{proof}

\begin{lemma}\label{lemma: conditional stoch boundedness}
For $Z_{\scriptscriptstyle{G}}$ as in \eqref{def: Z_G}, there exists a stochastically bounded sequence $\hat{M}_n$ such that:
\begin{align}\label{eq: conditional stoch boundedness 1}
	\Pi_n \left( \inf \{ t \in [-x n^{1/3}, (T-x)n^{1/3}] : Z_{\scriptscriptstyle{G}}(t) - a t \, \,  \text{is  minimal} \} \le \tilde{M}_n \mid Z_1, \dots, Z_n \right) \xrightarrow{G_0} 1 
\end{align}
\end{lemma}
\begin{proof}
First, for $T_{\scriptscriptstyle{G}}$ as in \eqref{def: T_G}, we need to show consistency of $T_{\scriptscriptstyle{G}}$ at $x$. By the switch relation $\forall \, \varepsilon > 0$,
\begin{align*} 
	&\Pi_n ( | T_{\scriptscriptstyle{G}}(F_{0}(x)) - x | > \varepsilon \mid Z_1, \dots, Z_n ) \\ 
	&\quad = \Pi_n ( T_{\scriptscriptstyle{G}} (F_0(x)) > x + \varepsilon \mid Z_1, \dots, Z_n ) + \Pi_n ( T_{\scriptscriptstyle{G}}(F_0(x)) < x - \varepsilon \mid Z_1, \dots, Z_n ) \\ 
	&\quad  \le \Pi_n ( \hat{F}_{\scriptscriptstyle{G}}(x+\varepsilon) < F_0(x) \mid Z_1, \dots, Z_n ) + \Pi_n ( \hat{F}_{\scriptscriptstyle{G}} (x-\varepsilon) > F_0(x) \mid Z_1, \dots, Z_n ) \xrightarrow{G_0} 0 \numberthis \label{eq: consistency of T_G}
\end{align*}
where in \eqref{eq: consistency of T_G} we use that $F_0$ is strictly increasing at $x$ by assumption and Lemma \ref{lemma: consistency} in the above which gives for every $\eta >0$: $\Pi_n(|\hat{F}_{\scriptscriptstyle{G}}(x \pm \varepsilon) - F_0(x \pm \varepsilon)| > \eta \mid Z_1, \dots, Z_n) \to 0$ under $G_0$.

The claim \eqref{eq: conditional stoch boundedness 1} is implied by the unconditional version of this convergence and this entails proving that the rate for $\left| T_{\scriptscriptstyle{G}} \left( F_0(x) \right) - x \right|$ is $n^{1/3}$. Because of Lemma \ref{lemma: consistency}, we can proceed by verifying the rest of the conditions of Theorem 3.2.5 in \cite{5} (or see Theorem 3.4.1 for a more general version).

The argmin process we are interested in is: $\tilde{t}_n := \operatorname{argmin}_t \{ Z_{\scriptscriptstyle{G}}(t) - a t \}$. Thus note that we can write it as: $\tilde{t}_n = n^{1/3} \operatorname{argmax}_t \{ - \mathbb{M}_n(t) \}$ where
\begin{align*} 
	\mathbb{M}_n(t) := \int (p(x+t-z) - p(x-z) - F_0(x)t) \, d(G - \mathbb{G}_n + G_0)(z) \\ + \int (p(x+t-z) - p(x-z)) \, d(\mathbb{G}_n - G_0)(z) - a t n^{-1/3}. 
\end{align*}
Moreover define the deterministic function;
\[ \mathbb{M}(t) := \int (p(x+t-z) - p(x-z) - F_0(x)t) \, dG_0(z) \]
Note $\mathbb{M}_n(0) = \mathbb{M}(0) = 0$. Moreover note that as a consequence of Taylor's theorem $\exists \, \delta_0 > 0: \forall t : |t| < \delta_0$:
\[ -\mathbb{M}(t) + \mathbb{M}(0) \lesssim -|t|^2 \]
Following Theorem 3.2.5. in \cite{5}, define $d(s,t) := |s-t|$. Consider the class of functions $m_t: \mathbb{R} \to \mathbb{R}$, for $\eta >0$:
\[ \mathcal{M}_\eta = \{ m_t : m_t(z) = p(x+t-z) - p(x-z) - F_0(x)t, \, \, 0 < t < \eta \} \]
We show that the envelope $M_\eta = \sup_{0<t<\eta} |m_t| $ of $\mathcal{M}_\eta$,
\[ \int M_\eta^2(z) g_0(z) dz \lesssim \int (p(x-z+\eta) - p(x-z))^2 g_0(z) \, dz + \eta^2 F_0(x) \]
where $\int (p(x-z+\eta) - p(x-z))^2 g_0(z) \, dz = p^2(0) g_0(x) \eta^2 + o(\eta^2)$. Thus for all $\eta$ small enough:
\[ \int M_\eta^2(z) g_0(z) \, dz \lesssim (p^2(0) g_0(x) + F_0(x)) \eta^2 \]

Similarly if we define the class of functions for $\eta >0$:
\[ \tilde{\mathcal{M}}_\eta = \{ \tilde{m}_t : \tilde{m}_t(z) = p(x+t-z) - p(x-z) : 0 < t < \eta \} \]
as aconsequence of the above, for the envelope $\tilde{M}_\eta := \sup_{0<t<\eta} |\tilde{m}_t|$ we have:
\[ \int \tilde{M}_\eta^2(z) g_0(z) dz \lesssim p^2(0) g_0(x) \eta^2. \]

By Proposition G.10 in \cite{4}, for independent \( Q \sim \text{DP}(\alpha) \), \( \mathbb{B}_n \sim \text{DP}(n \mathbb{G}_n) \), and \( V_n \sim \text{Be}(|\alpha|, n) \), the following representation in distribution of $G \sim \mathrm{DP}(\alpha + n \mathbb{G}_n)$: $G = V_n Q  + (1 - V_n) \mathbb{B}_n $, and thus
\begin{align}
\mathbb{M}_n(t) - \mathbb{M}(t) &= (1-V_n) \int \left( p(x+t-z) - p(x-z) - F_0(x)t \right) \, d(\mathbb{B}_n - \mathbb{G}_n)(z) \\
&\quad + V_n \int \left( p(x+t-z) - p(x-z) - F_0(x)t \right) \, d(Q-\mathbb{G}_n)(z) \\
&\quad + \int \left( p(x+t-z) - p(x-z) \right) d(\mathbb{G}_n - G_0)(z) - a t n^{-1/3} 
\end{align}
By applying the maximal inequality in Theorem 2.14.1 in \cite{5}
\begin{align*} 
	\ex \left[ \sup_{|t|<\delta} \sqrt{n} |\mathbb{M}_n - \mathbb{M}|(t) \right] &\lesssim J(1, M_\delta, \mathbb{L}_2) \sqrt{\int M_\delta^2 \, dG_0} + \sqrt{n} \frac{1}{n} \delta^\gamma + |a| \delta \sqrt{n} n^{-1/3} \\ 
	& \quad + J(1, \tilde{M}_\delta, \mathbb{L}_2) \sqrt{\int \tilde{M}_\delta^2 \, dG_0} \lesssim \delta + n^{-1/2} \delta^\gamma + a \delta n^{1/6} \\
	&\sim \delta(1 + a n^{1/6}) \end{align*}
Because $\delta$ must satisfy: $\delta \ge n^{-1/3}$. All together the above
proves that the rate for $\operatorname{argmax}_t \{-\mathbb{M}_n(t)\}$ is $\delta^{-1}$.
\end{proof}
\section{Review of the integral equation}\label{sec: review integral equation}

A very important question regarding integral equation \eqref{eq:1} is whether there exist a function $p$ satisfying it. In this section, we discuss conditions on the kernel $k$ that are sufficient for the existence of a solution $p$ to \eqref{eq:1}. This appendix is just a review based on the known results in \cite{104} and \cite{106} (no new results are given). We emphasize that the main results of this paper (Theorem \ref{thm: bayesian UQ} and Theorem \ref{thm: frequentist UQ}) only require $p$ to satisfy Condition \ref{cond:2}, which is a H\"{o}lder continuity condition on $p$ itself. The sufficient conditions discussed in this appendix, centering on the special case of \emph{decreasing} densities $k$, are therefore not necessary for the main results; they serve to illustrate a concrete and well-studied class of kernels for which Condition \ref{cond:2} can be verified. As demonstrated in the simulation study (Section 3), Condition \ref{cond:2} is also met by a wider variety of kernels that are not necessarily decreasing. Suppose:
\begin{condition}\label{cond on k}
$k: [0,\infty) \to [0,\infty)$ is a decreasing density and $0 < k(0) < \infty$.
\end{condition}
Without loss of generality we may take $k$ to be right continuous.
For a kernel satisfying Condition \ref{cond on k}, properties of the associated function $p$ solving equation \eqref{eq:1} can be obtained using some results from \emph{renewal theory} (see e.g.\ \cite{105}). Define the distribution function $J$ on $[0,\infty)$ by
\[ J(x) := 1 - \frac{k(x)}{k(0)}. \]
Note that $\mu_J = \int_0^\infty x \, dJ(x) = 1/k(0)$. Renewal theory assures the existence of a renewal function $m$ on $[0,\infty)$ satisfying the renewal equation
\[ m(x) = J(x) + \int_0^x m(x-y) \, dJ(y). \]
Integrating this equation with respect to $x$, we obtain
\begin{align}\label{eq:4a}
\mathbf{1} * m = (\mathbf{1} + m) * J. 
\end{align}
Writing
\begin{align}\label{eq:5a}
	p(x) = \frac{1+m(x)}{k(0)} = \mu_J (1+m(x)) \quad \text{for } x \ge 0, 
\end{align}
it follows that $p$ solves \eqref{eq:1}. Indeed, for $x > 0$,
\[ p * k(x) = \frac{1}{k(0)}(\mathbf{1}+m) * k(x) = (\mathbf{1}+m) * (\mathbf{1}-J)(x) = \mathbf{1} * \mathbf{1}(x) \]
by the definition of $J$ and \eqref{eq:4a}. 

\begin{lemma}[Lemma 1 in \cite{106}]
Let the density $k$ on $[0,\infty)$ satisfy Condition \ref{cond on k}. Then
\begin{enumerate}[label=\alph*)]
    \item there is a uniquely determined solution $p$ to (1) which is bounded on bounded intervals and vanishes on $(-\infty, 0)$
    \item writing $J^{*1} = J$ and $J^{*n} = J^{*(n-1)} * J$ for $n=2,3,\dots$, this uniquely determined $p$ on $[0,\infty)$ is given by
    \begin{align}\label{eq:6a}
    	p(x) = \frac{1}{k(0)} \left( 1 + \sum_{n=1}^\infty J^{*n}(x) \right) = \frac{1}{k(0)} \left( 1 + \sum_{n=1}^\infty \mathrm{P}\left(\sum_{j=1}^n T_j \le x \right) \right), 
    \end{align}
    where $T_1, T_2, \dots$ is an iid sequence of random variables such that $T_1$ has distribution function $J$.
    \item the function $p$ is nonnegative and nondecreasing on $[0,\infty)$.
    \item \[\lim_{x \to \infty} \frac{p(x)}{x} = 1. \]
\end{enumerate}
\end{lemma}
\begin{proof}[Proof from \cite{106}] 
All statements follow from relation (5) together with known results from renewal theory. See, for example, \cite{105} Theorem 10.1.11 for a), Lemma 10.1.7 for b) and Theorem 10.2.3 for d). Part c) is obvious from the definition of a renewal function.
\end{proof}

Because of Condition \ref{cond:2}, which we use to derive our main result Theorem \ref{thm: bayesian UQ}, one interesting question is to know conditions on $k$ under which $p$ is Lipschitz continuous on intervals $[0, T]$ for any bounded $T$. Functions $p$ associated with kernels satisfying Condition \ref{cond on k} do not have this property without further assumptions. It follows from \eqref{eq:6a} that a jump of $k$ at a point $x > 0$ causes infinitely many jumps in the function $p$. The condition on $k$ given below in Lemma \ref{lemma:1} is sufficient for $k$ to guarantee that $p$ is Lipschitz continuous on bounded intervals in $[0, \infty)$. It gives thus a sufficient condition for Condition \ref{cond:2} and thus also for Condition \ref{cond:1} to hold.

\begin{lemma}[Lemma 1 in \cite{74}] \label{lemma:1}
Let $0 < T < \infty$. Suppose the density $k$ can be written as
\[
k(x) = k(0) \left(1 + \int_0^x l(u) \, du \right), \quad x \in [0,T]
\]
for some bounded Borel measurable function $l: [0,T] \to \mathbb{R}$. Then the unique continuous (on $(0,T)$) solution $p$ of \eqref{eq:1} allows the representation
\[
p(x) = \frac{1}{k(0)} \left(1 + \int_0^x q(u)du \right), \quad x \in [0,T],
\]
where $q$ is a bounded Borel measurable function on $[0,T]$.
\end{lemma}
\begin{proof}[Proof from \cite{74}] 
Consider the type-II Volterra convolution equation
\[ q(x) + \int_0^x l(x-u)q(u) du = -l(x), \]
or, equivalently
\[ q + l*q = -l. \]
By Theorem 3.5 in Gripenberg, Londen and Staffans (1990), it follows that the solution $q$ of this equation uniquely exists. It is bounded and Borel measurable on $[0, T]$ whenever $l$ is. Now define $p = k(0)^{-1}(\ind + \ind*q)$ and observe
\[
k*p = k(0)(\ind + \ind*l) * \frac{\ind + \ind*q}{k(0)} = \ind*\ind + \ind*\ind*(q+l+q*l) = \ind*\ind.
\]
\qedhere
\end{proof}

\end{appendix}

\bibliographystyle{apalike}
\bibliography{mybiblio}

\end{document}